\documentclass[acmsmall,nonacm]{acmart}

\PassOptionsToPackage{dvipsnames}{xcolor}
\usepackage{graphicx}
\usepackage{xcolor}
\usepackage[ruled,vlined,linesnumbered]{algorithm2e}
\usepackage{algpseudocode}
\usepackage{tikz}
\usetikzlibrary{calc,decorations.markings}
\usepackage{makecell}
\usepackage{quiver}
\usepackage{amsmath}
\usepackage{enumitem}
\usepackage{mathtools}
\usepackage{thm-restate}
\usepackage{cleveref}
\DeclareGraphicsExtensions{.png,.pdf,.svg}

\SetKw{Let}{let}
\SetKw{Break}{break}
\SetKwInput{KwInput}{Input}
\SetKwInput{KwShared}{Shared}
\SetKwInput{KwInitial}{Local}
\SetKwInput{KwPattern}{write-read pattern  }
\SetKwInput{KwImmediately}{immediately}
\SetKwProg{Fn}{Function}{:}{}

\newcommand{\mynote}[3]{
    \fbox{\bfseries\sffamily\scriptsize#1}
{\small$\blacktriangleright$\textsf{\emph{\color{#3}{#2}}}$\blacktriangleleft$}}
\newcommand{\ignore}[1]{}

\newboolean{editing}
\setboolean{editing}{true}

\ifthenelse{
        \boolean{editing}
    }{
    \newcommand{\ps}[1]{\mynote{Pierre}{#1}{red}}
    \newcommand{\pk}[1]{\mynote{Petr}{#1}{cyan}}
    \newcommand{\gt}[1]{\mynote{Guillermo}{#1}{teal}}    
    }
    {
    \newcommand{\ps}[1]{}
    \newcommand{\pk}[1]{}
    \newcommand{\gt}[1]{}    
    }

\newcommand{\True}{\text{True}}
\newcommand{\False}{\text{False}}

\newcommand{\calC}{\mathcal C}

\newcommand{\ops}[1]{\textit{ops}(#1)}

\newcommand{\ccdot}{\!\cdot\!}

\newcommand{\comp}[1]{\textit{comp}(#1)}

\newcommand{\ret}{\mathsf{ret}}
\newcommand{\st}{\mathsf{st}}
\newcommand{\gca}{\text{GCA}}
\newcommand{\gcapropose}[2]{\gca_{#1}\text{.propose}(#2)}
\newcommand{\propose}[1]{\textit{propose}(#1)}
\newcommand{\dtrace}{O^*\!/\!\sim}
\usepackage{varwidth}
\usepackage[most]{tcolorbox}
\newsavebox{\algohlbox}

\newtcolorbox{codehighlight}{
  colback=yellow!12,
  colframe=yellow!12,
  boxrule=0pt,
  arc=1pt,
  boxsep=0pt,
  left=2pt,right=2pt,top=1pt,bottom=1pt,
  before skip=0pt,
  after skip=0pt,
}

\title{Conflict-Freedom as a Progress Condition {\small (Extended Version)}}
\titlenote{This is an extended version of a conference article. The conference version of this article appears in the proceedings of the
2026 ACM Symposium on Principles of Distributed Computing (PODC’26), July 6–10, 2026, Egham, United Kingdom, \url{https://doi.org/10.1145/3796701.3815956}.}
\makeatletter
\def\@ACM@checkaffil{}
\makeatother

\author{Petr Kuznetsov}
\orcid{0000-0003-1148-1228}
\affiliation{%
\institution{T\'el\'ecom Paris, Institut Polytechnique de Paris}
}
\email{petr.kuznetsov@telecom-paris.fr}

\author{Pierre Sutra}
\email{pierre.sutra@telecom-sudparis.eu}
\orcid{0000-0002-0573-2572}
\affiliation{%
  \institution{T\'el\'ecom SudParis, Institut Polytechnique de Paris}
}

\author{Guillermo Toyos-Marfurt}
\email{guillermo.toyos@telecom-paris.fr}
\orcid{0000-0002-6830-1649}
\affiliation{%
  \institution{T\'el\'ecom Paris, T\'el\'ecom SudParis, Institut Polytechnique de Paris}
}

\begin{document}

\begin{abstract}
An obstruction-free implementation guarantees progress to every operation that is given enough time to take steps in isolation.
But, as we show in this paper, the mere presence of concurrent operations alone does not have to prevent progress; only incomplete \emph{conflicting} (non-commuting) operations may do so.
This progress condition, that we call \emph{conflict-freedom}, is a natural generalization of obstruction-freedom that promises efficient implementations for objects exhibiting semantic commutativity. 
We show that, as with obstruction-freedom, every sequential object has a read-write conflict-free linearizable implementation.   
Our conflict-free universal construction is based on a novel generalization of the instrumental commit-adopt object, interesting in its own right.
\end{abstract}

\maketitle


\section{Introduction}

A major theme in distributed computing is the interaction between \emph{safety} and \emph{liveness}~\cite{safety-liveness-lamport,safety-liveness-alpern}. 
Intuitively, a safety property imposes restrictions on the outputs that a distributed system is allowed to produce, and a liveness property (also called a \emph{progress condition}) specifies when the outputs \emph{must} be produced, despite process failures and imperfect communication. 
The desired levels of safety and liveness translate into inherent implementation costs and model assumptions, discovering these is in the very spotlight of theoretical distributed computing.

The gold standard of safety is \emph{linearizability}~\cite{HW90-lin,attiya2004distributed}: a concurrent implementation of a data structure creates the illusion of an \emph{atomic} object where each operation takes effect instantaneously at some indivisible moment between the invocation of the operation and its response.
The most attractive progress condition is \emph{wait-freedom}~\cite{Her91}: every operation invoked by a process must return in a finite number \emph{of its own steps}, i.e., regardless of the behavior of concurrent processes.

Ensuring both linearizability and wait-freedom is notoriously expensive and, for many popular data structures, even impossible~\cite{flp,LA87,Her91}. 
As safety is often considered to be indispensable, multiple weaker liveness conditions, such as starvation-freedom, deadlock-freedom, lock-freedom and obstruction-freedom~\cite{HS11,artmp}, have been considered in the literature. 
In particular, lock-freedom (originally known as \emph{nonblockingness}~\cite{Her91}) stipulates that in \emph{any run}, \emph{at least one} correct process makes progress.
Even though lock-freedom allows us to  do things that cannot be done wait-free~\cite{lf-wf}, it is still too strong for many other things. 

The weaker criterion of \emph{obstruction-freedom}~\cite{of} has gained momentum, as it enables a read-write \emph{universal construction}, i.e., a read-write linearizable implementation of any object, given its sequential specification. 
An obstruction-free algorithm ensures progress in runs that are \emph{eventually step-contention-free}, i.e., a single process that eventually runs in isolation is guaranteed to complete its operations on the implemented object.  
In this sense, obstruction-freedom is a \emph{scheduler-dependent} progress condition~\cite{HS11}: it guarantees \emph{optimistic} progress in the hope that the scheduler will give enough time to every contending process to run in isolation.

In this paper, we introduce a natural generalization of obstruction-freedom that accounts 
for \emph{conflicts}.
Indeed, it is commonplace that the main difficulty in distributed computing is the need to \emph{synchronize} concurrent data accesses, i.e., to resolve race conditions by using consensus or strong synchronization primitives~\cite{ofc}.  
It is often argued that conflicts are rare in practice~\cite{DavidG16}, as many objects inherently offer a lot of parallelism at the data level, e.g., by composing their states into disjoint components that can be accessed independently.    
One can therefore seek an optimistic progress condition that ensures progress in \emph{eventually conflict-free runs}, i.e., under the assumption that no conflict (i.e., no pair of non-commuting operations) is encountered by active processes from some point on.
%
We call this criterion \emph{conflict-freedom} and provide it in two flavors: conflict-freedom that makes sure that \emph{every} active process makes progress in eventually conflict-free runs and \emph{weak} conflict-freedom that ensures progress for one of them.

Herlihy and Shavit~\cite{HS11} proposed a classification of progress conditions, based on the guarantees they provide and  the requirements they impose on the environment (Table~\ref{tab:progress-table}).
%
Conflict-freedom can therefore be seen as a relaxation of wait-freedom; weak conflict-freedom, a relaxation of lock-freedom.
Conflict-freedom is to weak conflict-freedom what wait-freedom is to lock-freedom: the former guarantees progress to every correct process, whereas the latter guarantees progress to at least one.
Both properties are, in general, stronger than obstruction-freedom, as they also provide some progress in contended but conflict-free scenarios. In the extreme case, however, when every pair of operations conflicts, both properties coincide with obstruction-freedom. 
When all operations commute, conflict-freedom coincides with wait-freedom, and weak conflict-freedom with lock-freedom.
\begin{table}[htp]
     \begin{center}
  \begin{tabular}{||c|c|c|c|c||}
  \hline
   & \makecell[c]{independent \\ nonblocking} & \makecell[c]{\bf conflict-dependent\\ \bf nonblocking} & \makecell[c]{dependent \\ nonblocking} & \makecell[c]{dependent \\ blocking} \\
  \hline
\makecell[c]{maximal \\progress}  & wait-freedom & \bf conflict-freedom & obstruction-freedom & starvation-freedom  \\
\hline
\makecell[c]{minimal \\progress} & lock-freedom & \bf \makecell[c]{weak \\conflict-freedom} & ? & deadlock-freedom \\
\hline
 \end{tabular}
\end{center}
  \caption{``Periodic table'' of progress conditions~\cite{HS11}, extended. Informally, a property is \emph{nonblocking} if it guarantees progress even if some processes halt; a property is \emph{independent} if it guarantees progress, regardless of how the steps of the concurrent processes are scheduled; \emph{maximal progress} ensures that every process (in a selected group) makes progress, while \emph{minimal progress} only ensures progress for \emph{at least one} selected process. The ephemeral property of \emph{clash-freedom} (the ``Einsteinium of progress conditions''~\cite{HS11}) is marked with $?$ in the table, as it is not used in practice.}
     \label{tab:progress-table}
     \vspace{-1em}
\end{table}

 
In this paper, we describe read-write universal constructions that provide weak conflict-freedom and conflict-freedom.
Inspired by obstruction-free consensus algorithms based on the \emph{commit-adopt} abstraction~\cite{paxos,Gaf98}, we present our constructions in a modular way, using our novel \emph{generalized} commit-adopt (GCA) abstraction that accounts for conflicts.

Intuitively, GCA can be seen as an optimistic one-shot \emph{conflict resolver} that takes inputs from a set equipped with a binary \emph{compatibility relation} capturing conflicts, and returns a compatible set of outputs. 
Moreover, if the inputs are conflict-free, one of them is \emph{committed} and the other outputs will be \emph{extensions} of it.  
We present a read-write GCA implementation and show how a weakly conflict-free universal construction can be built by simply \emph{chaining} GCA objects: the outputs of one GCA object are used as inputs to the next one.
Finally, we present a conflict-free universal construction that combines the chaining idea of the weakly conflict-free construction with a helping mechanism.  

To sum up, our contributions are:
\begin{enumerate}
    \item the new progress criteria of weak conflict-freedom and conflict-freedom;   
    \item a new abstraction, generalized commit-adopt (GCA), interesting in its own right, and its read-write implementation;
    \item a simple weakly conflict-free universal construction from a chain of GCA objects; 
    \item a conflict-free universal construction using a chain of GCA objects and a read-write helping mechanism. 
\end{enumerate}

The rest of the paper is organized as follows. Section~\ref{sec:related} discusses related work. 
Section~\ref{sec:model-cf} recalls basic model definitions and introduces the notion of (weak) conflict-freedom. 
Section~\ref{sec:ucwcf} introduces the GCA abstraction and presents our weakly conflict-free universal construction.
%
%
In Section~\ref{sec:GCAimp}, we describe a read-write GCA implementation. 
In Section~\ref{sec:univconstructionv2}, we show how our weakly conflict-free universal construction can be turned into a conflict-free one using a helping mechanism.   
We conclude the paper with ramifications and open questions in Section~\ref{sec:conclusion}.  
Detailed proofs are delegated to Appendix~\ref{sec:proofs}.

The model and results presented in this work have been formalized in Lean~4, available at~\cite{CFLean}.
In particular, every numbered lemma, theorem, and proposition is accompanied by a corresponding machine-checked proof.
The formalization assumes wait-free linearizable snapshot objects.
The Lean formalizations and proofs were developed with the assistance of Claude Opus~5.5. 
We provide these formalizations as supplementary material, offering an additional layer of machine-checked verification for the results of this work. 
They complement the natural-language proofs in the paper, which remain the primary presentation of the mathematical arguments.
%
%

\section{Related Work}
\label{sec:related}



The commit-adopt abstraction was introduced by Gafni~\cite{Gaf98}.
Herlihy~\cite{Her91} proposed the notion of universal construction and presented the first wait-free consensus-based algorithm for it.  
While consensus is necessary for wait-freedom, one can build an obstruction-free universal construction~\cite{of} that only guarantees progress in the absence of step contention.
%
%
The space and time complexity of obstruction-free algorithms remains an active research topic~\cite{ofc,Zhu21,Ovens24}. 

%

%

The idea of commutativity finds its formal foundation in the theory of Mazurkiewicz traces~\cite{aalbersberg1988theory, Mazurkiewicz84}.
Traces permit the generalization of universal constructions (UC) and state-machine replication (SMR), a distributed analogue of a (shared-memory) UC, to account for commutativity.  
%
%
%
%
Generalized Paxos~\cite{lamport2005generalized} is a seminal example of using traces (termed \emph{command histories}) and commutativity to build an SMR protocol allowing replicas to process commands in different orders, as long as they commute.
Generalized Paxos implements the \emph{ballot array}, a shared-memory abstraction that associates a monotonically growing trace with each process and ballot.
As in our work, Generalized Paxos uses least upper bounds (LUBs) and greatest lower bounds (GLBs) to merge entries in the ballot array.
In a system of $n=2f+1$ processes, if at most $f/2$ processes are faulty, the ballot array enables progress in (asynchronous) eventually conflict-free runs. 
%

Building on this idea, Egalitarian Paxos~\cite{moraru13epaxos, RyabininGS25} eliminates the leader role and constructs a dependency graph over conflicting operations, enabling decisions in fewer communication steps.
A general framework for dependency-based SMR protocols is presented in~\cite{rezende20leaderless}. 
%

A related abstraction is \emph{generic broadcast}~\cite{Schiper99Generic,Zielinski05Broadcast}. 
Unlike atomic broadcast~\cite{cristian1995atomic}, which enforces total order on all message deliveries, generic broadcast requires only conflicting messages to be ordered consistently, which makes this abstraction equivalent to generalized SMR. 
%
%

The observation that commutative operations can be executed in any order has motivated a wide range of optimizations (to name a few, \cite{Malek05QU,Cowling06HQ,Suri21Basil}), but designing practical systems that fully exploit commutativity remains challenging~\cite{whittaker2021sok}. 
Our work appears to be the first to formally define a \emph{semantics-aware} liveness  property that accounts for application-specific conflicts. 

\ignore{
Despite theoretical advances, designing practical SMR systems that fully exploit commutativity remain challenging~\cite{whittaker2021sok}.
Most systems adopt conservative definitions of commutativity, aligned with specific object semantics (e.g., disjoint read-write sets).
Early BFT SMR systems incorporating commutativity include~\cite{Cowling06HQ, Malek05QU}, which use object-level conflicts to reduce coordination.
These ideas are developed in~\cite{Suri21Basil},  which introduces fast-path execution using stronger quorum assumptions ($n = 5f + 1$).
%
}

\section{Conflict-Freedom} 
\label{sec:model-cf}

\subsection{Preliminaries}
\label{sec:preliminaries}

We assume the standard model of \citet{HW90-lin}, where processes communicate using shared memory to implement a linearizable object.
Below, we recall the main features of this model.

\paragraph{Objects.}
A (high-level) object is defined by a sequential data type given by a tuple $(Q,q_0,O,R,\sigma)$, where $Q$ is a set of \emph{states}, $q_0\in Q$ is the \emph{initial state}, $O$ is a set of \emph{operations}, $R$ is a set of \emph{responses}, and $\sigma:\;O\times Q \rightarrow R\times Q$ is a \emph{transition function} specifying the response and resulting state of each operation in each state.    
Given a state $q \in Q$, operations $a,b \in O$ \emph{commute in $q$} if there exist responses $r_a,r_b \in R$ and states $q_1,q_2,q' \in Q$ such that $\sigma(a,q)=(r_a,q_1)$, $ \sigma(b,q_1)=(r_b,q')$, $\sigma(b,q)=(r_b,q_2)$ and $\sigma(a,q_2)=(r_a,q')$.
That is, executing $a$ then $b$ from $q$ yields the same responses and final state as executing $b$ then $a$.
Commutativity induces a binary, symmetric, \emph{conflict} relation $\asymp\subseteq O\times O$, where operations $a,b \in O$ are \emph{conflicting}, written $a \asymp b$, if there exists a state $q \in Q$ in which $a$ and~$b$ do not commute.

\paragraph{Algorithms and executions.}
We consider an asynchronous shared-memory model with $n$ processes, $\Pi = \{i : 0 < i \leq n\}$, that communicate through atomic read-write registers, in order to implement a shared object.
Each process is assigned a deterministic \emph{algorithm}.
A \emph{step} of a process consists of a local computation followed by at most one
atomic (read or write) operation on a register. 
An \emph{execution fragment} is a (finite or infinite) sequence of steps
taken by processes according to their algorithms.
An \emph{execution} is an execution fragment that starts from the algorithm's initial configuration of each process.
In an infinite execution, a process is \emph{correct} if it either takes infinitely many steps or has no pending operation; otherwise it is faulty.
There can be any number of failures during an execution.

\paragraph{Histories and linearizability.}
%
Each execution induces a \emph{history}, which is a sequence of operation invocation and response events by the processes on the shared object.
An invocation of an operation $o$ on the shared object by a process $i$ creates an \emph{operation instance}, $\Phi = (o,i)$.
An invocation \emph{completes} if it has a matching response in the history; otherwise it is \emph{pending}.
%
%
Two operation instances are \emph{concurrent} if their execution intervals overlap, that is, neither completes before the other is invoked.
%
%
%
Our goal is to build a \emph{linearizable} implementation of an object that satisfies its sequential specification~\cite{HW90-lin,attiya2004distributed}.
Informally, a history is linearizable if all its complete operations and some of its incomplete ones can be assigned \emph{linearization points}, indivisible moments of time within their intervals, so that these operations put in the order of their linearization points constitute a correct sequential history of the implemented object.
An implementation is linearizable if every history it produces is linearizable.

\paragraph{Liveness.}
An implementation is \emph{wait-free} (resp. \emph{lock-free}) if in each of its infinite executions, \emph{every} (resp. \emph{at least one}) correct process that takes infinitely many steps completes each of its operations.
%
%
An operation instance $\Phi$ is \emph{eventually step-contention free} in an execution $\alpha$ if either $\Phi$ completes in $\alpha$, or there exists a suffix of $\alpha$ in which only
$\Phi$ takes steps (i.e., from some point on in $\alpha$, $\Phi$ runs solo).
An implementation is obstruction-free if in each of its infinite executions, any operation instance $\Phi$ invoked by a correct process completes in $\alpha$ whenever $\Phi$ is eventually step-contention free~\cite{of,ofc}. 

\subsection{Conflict-Free Progress: Definition} \label{sec:cf-def}

%
%
%

An infinite execution $\alpha$ is \emph{eventually conflict-free} if there exists a suffix of $\alpha$ in which no two concurrent operation instances $\Phi = (o,i)$ and $\Phi' = (o',j)$ satisfy $o \asymp o'$.
Conflict-freedom extends obstruction-freedom by also ensuring progress in such executions:

\begin{definition}[Conflict-freedom]
An implementation is \textit{conflict-free} if, in every infinite execution, every operation invoked by a correct process completes whenever the operation is eventually step-contention-free \textbf{or} the execution is eventually conflict-free.
\end{definition}

\begin{definition}[Weak conflict-freedom]
An implementation is \textit{weakly conflict-free} if, (1) in every infinite execution, every operation invoked by a correct process completes whenever it is eventually step-contention-free, \textbf{and} (2) if the execution is eventually conflict-free, then \emph{some} correct process that takes infinitely many steps completes all of its operations.
\end{definition}

%

\begin{figure}
\centering
\resizebox{\columnwidth}{!}{%
\begin{tikzpicture}[x=0.88cm,y=1.22cm,
  proc/.style={font=\small,anchor=east},
  timeline/.style={thick,->},
  br/.style={font=\small,inner sep=0pt},
  lab/.style={font=\scriptsize,inner sep=0pt}
]

\def\xStart{0.05}
\def\xEnd{7.60}
\def\panelShift{7.35}

\pgfmathsetmacro{\dotRatio}{0.88/1.22}
\def\dotR{0.05}
\def\dotLW{0.6pt}

\newcommand{\op}[4]{%
  \node[br] at (#2,#1) {$[\,$};
  \node[lab,anchor=east] at (#2+0.155,#1+0.24) {#4};
  \ifx&#3&%
    \draw[thick] (#2+0.20,#1) -- (\xEnd-0.25,#1);
  \else
    \draw[thick] (#2+0.20,#1) -- (#3-0.20,#1);
    \node[br] at (#3,#1) {$\,]$};
  \fi
}

\newcommand{\uniformsteps}[4]{%
\pgfmathsetmacro{\step}{(#3-#2)/(#4+1)}
\foreach \k in {1,...,#4} {%
  \pgfmathsetmacro{\xx}{#2+\k*\step}
  \filldraw[line width=\dotLW] (\xx,#1) ellipse
    [x radius=\dotR, y radius=\dotR*\dotRatio];
}%
}

\def\lblX{-0.05}

\node[proc] at (\lblX,3) {$p_1$};
\node[proc] at (\lblX,2) {$p_2$};
\node[proc] at (\lblX,1) {$p_3$};

\pgfmathsetmacro{\midXBase}{(\xEnd + (\panelShift+\xStart))/2.0}
\def\sepNudge{0.275}
\pgfmathsetmacro{\midX}{\midXBase+\sepNudge}
\draw[thin] (\midX,0.65) -- (\midX,3.35);

\node[font=\scriptsize,anchor=east] at (\lblX,3.55) {(a)};
\node[font=\scriptsize,anchor=west] at (\midX+0.14,3.55) {(b)};

\begin{scope}
  \draw[timeline] (\xStart,3) -- (\xEnd,3);
  \draw[timeline] (\xStart,2) -- (\xEnd,2);
  \draw[timeline] (\xStart,1) -- (\xEnd,1);

  \op{3}{0.75}{2.35}{\textit{read}()}
  \op{3}{5.95}{7.10}{\textit{inc}()}

  \op{2}{1.25}{3.70}{\textit{inc}()}
  \op{2}{4.00}{6.55}{\textit{dec}()}

  \op{1}{1.05}{3.30}{\textit{dec}()}
  \op{1}{3.85}{5.05}{\textit{inc}()}
  \op{1}{5.35}{6.55}{\textit{dec}()}

  \draw[dashed,thick] (2.45,0.65) -- (2.45,3.35);

  \uniformsteps{3}{0.75}{2.35}{4}
  \uniformsteps{3}{5.95}{7.10}{3}

  \uniformsteps{2}{1.25}{3.70}{5}
  \uniformsteps{2}{4.00}{6.55}{5}

  \uniformsteps{1}{1.05}{3.30}{4}
  \uniformsteps{1}{3.85}{5.05}{3}
  \uniformsteps{1}{5.35}{6.55}{3}
\end{scope}

\begin{scope}[xshift=\panelShift cm]
  \draw[timeline] (\xStart,3) -- (\xEnd,3);
  \draw[timeline] (\xStart,2) -- (\xEnd,2);
  \draw[timeline] (\xStart,1) -- (\xEnd,1);

  \op{3}{0.75}{2.35}{\textit{read}()}
  \op{3}{5.20}{}{\textit{inc}()}

  \op{2}{1.25}{3.70}{\textit{inc}()}
  \op{2}{4.00}{5.55}{\textit{dec}()}
  \op{2}{5.95}{7.10}{\textit{inc}()}

  \draw[dashed,thick] (2.45,0.65) -- (2.45,3.35);

  \op{1}{3.25}{}{\textit{inc}()}

  \uniformsteps{3}{0.75}{2.35}{4}
  \uniformsteps{3}{5.20}{7.10}{4}

  \uniformsteps{2}{1.25}{3.70}{5}
  \uniformsteps{2}{4.00}{5.55}{4}
  \uniformsteps{2}{5.95}{7.10}{3}

  \uniformsteps{1}{3.25}{7.10}{5}
\end{scope}

\end{tikzpicture}%
}
\caption{
Executions on a shared counter under (a) conflict-freedom and (b) weak conflict-freedom.
Brackets denote invocations/completions and dots mark steps.
A read conflicts with an increment/decrement, whereas two updates commute.
After the dashed line no conflicting operations overlap.
In (a) all updates complete, whereas in (b) only $p_2$ makes progress and the increments of $p_1$ and $p_3$ remain pending.
}
\Description{
Executions on a shared counter under (a) conflict-freedom and (b) weak conflict-freedom.
Execution (a): $p_1$ invokes a read and then an increment; $p_2$ invokes an increment and then a decrement; $p_3$ invokes a decrement, an increment and then a decrement.
The read of $p_1$, the increment of $p_2$, and the decrement of $p_3$ are concurrent. 
No further conflicting operations overlap after these operations complete.
Execution (b): $p_1$ invokes a read and then an increment; $p_2$ invokes an increment, a decrement and then an increment; $p_3$ invokes an increment.
The read of $p_1$, and the first increment of $p_2$ are concurrent.
No further conflicting operations overlap after these operations complete.
In (a) all updates complete, whereas in (b) only $p_2$ makes progress and the increments of $p_1$ and $p_3$ remain pending.}
\label{fig:cf-vs-wcf}
\end{figure}
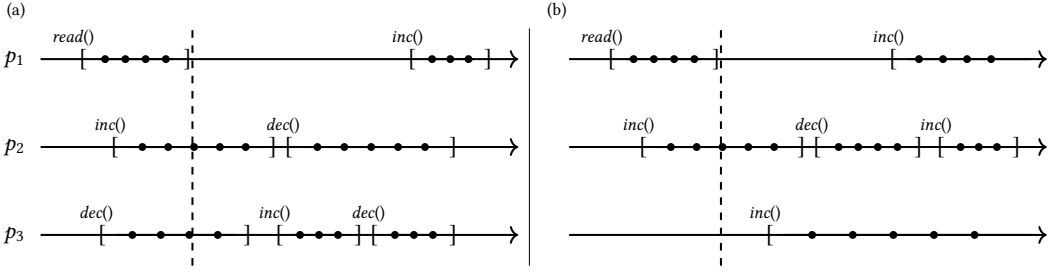

To build intuition, consider the executions in Figure~\ref{fig:cf-vs-wcf}, where three processes invoke operations on a shared counter storing an integer value and supporting $\textit{read}()$, $\textit{inc}()$ (increment), and $\textit{dec}()$ (decrement) operations.
Reads conflict with update operations ($\textit{inc}$ and $\textit{dec}$), whereas any two updates commute or any two reads commute.
In these executions, $p_1$ invokes a read operation that eventually completes and subsequently invokes increment operations, while $p_2$ and $p_3$ invoke arbitrarily many increments and decrements.
At some point, the executions turn conflict-free, although operations continue to take steps concurrently. 
Conflict-freedom guarantees that all updates complete in this setting, whereas weak conflict-freedom guarantees progress only for some process.
In an obstruction-free algorithm, no progress can be ensured here, as the execution never becomes step-contention-free.

\begin{restatable}{proposition}{ProgressHierarchy}
\label{prop:progress-hierarchy}
The progress conditions form the following hierarchy:
\\
$
\text{wait-freedom} \implies \text{conflict-freedom} \implies \text{weak conflict-freedom} \implies \text{obstruction-freedom}
$
\end{restatable}
\ignore{
\begin{proof}[Proof (Sketch)]
Wait-freedom guarantees completion of every operation in every infinite execution, and therefore implies conflict-freedom. 
Conflict-freedom implies weak conflict-freedom since it guarantees progress for every operation in eventually conflict-free executions, rather than for at least one process.
Finally, weak conflict-freedom ensures that every eventually step-contention-free operation completes, which is exactly the guarantee of obstruction-freedom.
\end{proof}
}

Conflict-freedom and weak conflict-freedom are parameterized by the conflict relation induced by the object. 
The following proposition characterizes the two degenerate conflict relations under which these notions reduce to classical progress conditions. 

\begin{restatable}{proposition}{CFDegenerate}
\label{prop:cf-degenerate}
Let $(Q,q_0,O,R,\sigma)$ be an object whose conflict relation is $\asymp$.
\begin{itemize}
    \item If $\asymp = O\!\times\!O$, then conflict-freedom and weak conflict-freedom are equivalent to obstruction-freedom.
    \item If $\asymp = \emptyset$, then conflict-freedom is equivalent to wait-freedom, and weak conflict-freedom is equivalent to lock-freedom.
\end{itemize}
\end{restatable}
%
%
%

\section{A Conflict-Free Universal Construction} \label{sec:ucwcf}
We show that every sequential object admits a read–write implementation that is both linearizable and conflict-free.
Universal constructions transform a sequential object specification into a linearizable shared implementation satisfying a given progress condition~\cite{Her91}. 
Wait-free and lock-free read–write universal constructions are impossible~\cite{flp,LA87,Her91}, whereas obstruction-free constructions are known to exist~\cite{of,ofc}.

The main challenge in implementing conflict-freedom is to guarantee progress without imposing a total order on concurrent operations. 
In eventually conflict-free executions, non-conflicting operations must be able to complete without agreeing on a relative order, whereas conflicting operations must still be ordered consistently to preserve linearizability.
We address this challenge using \textit{traces}~\cite{aalbersberg1988theory,Mazurkiewicz84}, which encode the ordering constraints induced by conflicts; we begin by introducing the trace model underlying the construction.

\subsection{Schedules and Traces} \label{sec:traces}

A \textit{schedule} is a finite sequence of operations $s \in O^*$.
Let $\epsilon$ denote the empty schedule.
For $o \in O$, let $|s|_o$ denote the multiplicity of $o$ in $s$; for  $i \in [1,|s|_o]$, let $o^{(i)}$ denote its $i$-th occurrence.
Let $\ops s$ denote the multiset over $O$ defined by $s$.
Let $\prec_s$ be the total order on occurrences induced by $s$.
For $k \in [1,|s|]$, we write $s[k]$ for the $k$-th operation in $s$.
Given schedules $x,y\in O^*$, let $x\cdot y$ denote their concatenation.

%
%
We say that two schedules $s,s' \in O^*$ are \emph{equivalent}, and write $s \sim s'$ if:
\begin{enumerate}
\item[(i)] $\ops{s}=\ops{s'}$ and  
\item[(ii)] For every pair of occurrences $a^{(i)},b^{(j)}$ in $s$, if $a \asymp b$ and $a^{(i)}\!\prec_s\! b^{(j)}$, then $a^{(i)}\!\prec_{s'}\!b^{(j)}$. 
\end{enumerate}
Informally, two schedules are considered equivalent if they contain the same operations, and any precedence between conflicting operation occurrences is preserved.

The relation $\sim$ is symmetric, reflexive, and transitive and, thus, induces a \emph{quotient space} on $O^*$, denoted $O^*/\sim$ (i.e., a set of equivalence classes in $O^*$).
The equivalence class of a schedule $s$ is denoted by $[s]$ and we call it a \emph{trace}.
Given two schedules, $s$ and $s'$, we define trace concatenation as $[s] \cdot [s'] = [s \cdot s']$.
We can verify that $\sim$ is a congruence relation over the monoid $(O^*,\cdot)$.
As such, concatenation is well-defined over traces and $(O^*/\sim,\cdot)$ is a quotient monoid.

This construction corresponds to Mazurkiewicz traces~\cite{aalbersberg1988theory,Mazurkiewicz84}.
Our conflict relation acts as a dependency relation, and equivalence classes of schedules correspond to traces, where commuting operations may be freely reordered.

A schedule $s$ is a representative of a trace $t$ if $[s]=t$; we write $s \in t$.
Let $\epsilon$ denote the empty trace.
The length of a trace $t$, denoted $|t|$ is the length of its representatives. 
All subsequent definitions and operations apply equally to schedules and traces.
%
We say that $x$ is \emph{a prefix of} $y$ (or that $y$ is \emph{an extension of} $x$), and write $x \leq y$ if there exists $z$ such that $y = x \cdot z$. 
If $z \neq \epsilon$, we say that $x$ is a strict prefix of $y$, written $x < y$.
Given a set $S\neq\emptyset$, we write the \emph{greatest lower bound (GLB) of} $S$, denoted $\bigsqcap S$, as the \emph{largest} $e$ such that $e \leq x$ for all $x \in S$.
Note that $\bigsqcap S$ can be the empty trace (schedule) $\epsilon$. 
We say that a set $S$ is \emph{comparable} if for all $x,y \in S$, either $x \leq y$ or $y \leq x$. 
A set $S$ is \emph{compatible}, denoted $\comp{S}$, if there exists $z$ such that for all $x \in S$, $x \leq z$. 
For a compatible set $S$, we define the \emph{least upper bound (LUB) of} $S$, denoted $\bigsqcup S$, as the \emph{smallest} $e$ such that $\forall x \in S, x \leq e$.
It is shown in~\cite{aalbersberg1988theory} that every nonempty set of traces admits a greatest lower bound (GLB), and that a least upper bound exists whenever the set is compatible.
Moreover, these bounds are built using the operations in $S$.
That is, $\ops{\bigsqcap S} \subseteq \bigcap_{s \in S} \ops{s}$  and $\ops{\bigsqcup S} \subseteq \bigcup_{s \in S} \ops{s}$.
The prefix order makes $(O^*/\sim,\leq)$ a \textit{bounded complete meet-semilattice}.
The GCA abstraction (Section~\ref{sec:gca}) is defined over this structure.

Traces represent the order in which operations are executed in the state machine.
We lift $\sigma$ to schedules by defining $\sigma^* : O^* \times Q \to (R \times Q)^*$ inductively:
\[
\begin{aligned}
&\sigma^*(\epsilon, q) = \epsilon, \\
&\sigma^*(o \ccdot s, q) = (r, q') \cdot \sigma^*(s, q'),
\qquad \text{where } (r,q') = \sigma(o,q).
\end{aligned}
\]

Schedules are interpreted starting from the initial state $q_0$.
Let $o^{(i)}$ denote the $i$-th occurrence of $o$ in a schedule $s \in O^*$. 
We define the return value of $o^{(i)}$ in $s$ as $\ret^*(o^{(i)},s) = \ret(\sigma^*(s,q_0)[k])$, where $s[k]$ is the $i$-th occurrence of $o^{(i)}$ in $s$.

\begin{restatable}{lemma}{RetTraceEq}
\label{lemma:ret-trace-eq}
Let $s,t\in O^*$ be two schedules such that $s\sim t$. Then for any occurrence of an operation $o^{(i)}$ in $s$, $ \ret^*(o^{(i)},s)=\ret^*(o^{(i)},t)$.
\end{restatable}

Lemma~\ref{lemma:ret-trace-eq} implies that the return value of every operation occurrence is invariant within an equivalence class.
%
%
Accordingly, we overload $\ret^*$ to traces by defining $\ret^*(o^{(i)},[s])=\ret^*(o^{(i)},\hat s)$ for any representative $\hat s \in [s]$ and operation occurrence $o^{(i)} \in \ops{[s]}$; this is well-defined.

\begin{restatable}{lemma}{RetTracePrefix}
\label{lemma:ret-trace-prefix}
Let $s,t \in O^*/\!\sim$, if $s \leq t$ then for every occurrence $o^{(i)}$ in $s$, $\ret^*(o^{(i)},s) = \ret^*(o^{(i)},t)$.
\end{restatable}

To illustrate this claim, consider a counter initialized to $0$ and the trace $t=[inc \ccdot dec \ccdot read]$.
We have that $\ret^*(read^{(1)},t)=0$. 
Since increments and decrements commute, every representative of the trace $t$ yields the same return value for the read.
Now consider an extension $t' = t \ccdot [inc \ccdot inc \ccdot read]$.
Although additional operations are appended, the return value of the first read remains unchanged:
$\ret^*(read^{(1)}, t) = \ret^*(read^{(1)}, t') = 0$. 
Thus, return values are invariant both across representatives of a trace and under trace extension, as captured by Lemmas~\ref{lemma:ret-trace-eq} and~\ref{lemma:ret-trace-prefix}.
Proofs for these two lemmas are in Appendix~\ref{sec:proofs}.

\subsection{Generalized Commit-Adopt} \label{sec:gca}
We introduce \emph{Generalized Commit-Adopt} (GCA), the main synchronization primitive underlying our universal construction.
GCA takes a trace as an input and returns, as an output, a trace with a boolean value that indicates whether the returned trace is \emph{committed} or just \emph{adopted}.
%
If all input traces are compatible, then an extension of at least one of them is \emph{committed}.
Moreover, if a trace is committed,  then all returned traces extend it.

\ignore{
Informally, each process proposes an element, and every returned element extends the committed one.
When the proposals are compatible, at least one process commits to an extension of its proposal, while the others adopt elements that extend it.
Committed elements thus act as synchronization points from which executions can safely progress.
We now formalize the abstraction.
}

Formally, GCA is a one-shot object defined over the trace domain $(\dtrace,\leq)$ induced by the conflict relation of the sequential specification.
%
%
The object exports a single operation $\textit{propose}\!:\! {(\dtrace)} \rightarrow {(\dtrace)} \times \{\True,\False\}$.
Each process $i$ invokes $\propose{s_i}$, where $s_i \in \dtrace$, and outputs $(t_i, c_i)$, where $t_i \in \dtrace$ and $c_i \in \{\True, \False\}$ indicates whether process $i$ \emph{commits} $t_i$ or \emph{adopts} $t_i$, respectively.

We say that a process is a \emph{participant} if it invokes the $\propose{}$ operation. 
The abstraction guarantees that every correct participant eventually returns. 
Let $P$ denote the set of participants, and let $P_r \subseteq P$ be the participants whose calls return in a given execution.
The following properties hold for every execution:

\begin{enumerate}
    \item \textbf{Validity.} Output traces contain only input operations.
    $\forall i \in P_r,~ ops(t_i) \subseteq \bigcup_{j\in P} ops(s_j)$
    \item \textbf{Adoption.} A committed trace is extended by every output trace.\\
    $\forall i \in P_r,~ c_i=\True \implies \forall j \in P_r, t_i \leq t_j$
    \item \textbf{Commitment.} If all input traces are compatible, then some process commits an extension of its input trace.  $(\comp{\bigcup_{i\in P}\{s_i\}}) \land (P=P_r\neq\emptyset)\! \implies\! \exists j\in P_r, (s_j\!\leq\!t_j) \land (c_j\!=\!\True)$

    \item \textbf{Convergence.} Output traces are mutually compatible.
    $\comp{\bigcup_{i \in P_r}\{t_i\}}$

    \item \textbf{Common Prefix.} Output traces preserve the common prefix of the input traces. \\
    $\forall i \in P_r,~ \sqcap_{j\in P}{~s_j} \leq t_i$
    
    \item \textbf{Weak Agreement.} If all input traces are equal, then no process adopts. \\
    $\forall i,j \in P, s_i=s_j \implies \forall i \in P_r : c_i = \True$.
\end{enumerate}

These properties guarantee that the output traces are compatible, extend any common prefix, and are built from the content of the input traces. 
Notice also that, by Weak Agreement, if all processes input the same trace, they must commit it.

At the heart of GCA is the interplay between Commitment and Adoption.
Commitment requires that whenever input traces are compatible and every process returns from its call, some process commits to an extension of its input.
Adoption forces all output traces to extend any committed one, thereby establishing a common synchronization point.

GCA lifts the standard commit-adopt (CA) abstraction~\cite{Gaf98} from values to traces.
In CA, processes propose values from a set $V$ and return a pair $(c,v)$, where $c\in\{\True,\False\}$ denotes commit or adopt and $v \in V$ is the returned value.
CA guarantees that if some process commits a value $v$, then every process commits or adopts $v$.
This corresponds to a special case of GCA with $O=V$ and the conflict relation $\asymp=V\!\times\!V$.
In detail, processes propose singleton traces $[v]$.
Since no operations commute, if some process commits a trace $t \neq \epsilon$, then by Adoption every output trace extends $t$, thereby ensuring that all processes agree on its first operation.
Given a GCA response $(t,c)$, a process returns the first operation in $t$ and $c$; if $t=\epsilon$ the process adopts its own input.

In Section~\ref{sec:GCAimp}, we describe a read-write GCA implementation. 

\subsection{Weakly Conflict-Free Universal Construction Using GCA} \label{sec:weak_uc}

Our weakly conflict-free universal construction (Algorithm~\ref{alg:UnivConstructionGCA}) follows a classical pattern, where state transitions of the implemented sequential object are computed using a series of ``agreement'' objects.
%
Intuitively, every newly invoked operation is proposed to the ``next'' agreement object in order.
If the operation is ``committed'' by the agreement object, it determines the next transition of the implemented object.
Otherwise, the operation is proposed again to the next agreement object.
We extend this pattern, by using GCAs as agreement objects and submitting growing traces of commands instead of individual operations.

\begin{algorithm}
\DontPrintSemicolon
\SetAlgoLined
\KwInput{$A= (Q,q_0,O,R,\sigma)$ \hfill $\triangleright$ Sequential object specification.}
\KwShared{$\{\gca_k(C^*\!/\!\sim,\leq) : k \geq 1\}$ \hfill $\triangleright$ Sequence of GCA objects.}
\KwShared{$S[1,\ldots, n] = (0,\epsilon)$ \hfill $\triangleright$ Array of atomic registers.}
\KwInitial{$r = 0;\; seq=0;\; s=\epsilon$}
\Fn{invoke(op)}{
    \Let $c = \False$ \label{line:uc:1} \\
    $seq \gets seq + 1$ \label{line:uc:2} \\
    $cmd \gets (op,i,seq)$ \label{line:uc:3} \\
    $(r,s) \gets \max_{r'} \{(r',s')=S[j] : j \in [1,n]\}$ \label{line:uc:4} \\
    \While{$(cmd \not\in \ops s) \lor (c=\False)$ \label{line:uc:5}}{
        $r \gets r+1$ \label{line:uc:6} \\
        \If{$cmd \notin \ops s$ \label{line:uc:7}}{
            $s \gets s \cdot cmd$ \label{line:uc:8} \\
        }
        $(s,c) \gets \gcapropose{r}{s}$ \label{line:uc:9} \\
    }
    $S[i] \gets (r,s)$ \label{line:uc:10} \\
    \Return\ $\ret^*(cmd,s)$ \label{line:uc:11} \\
}
\caption{Weakly conflict-free universal construction algorithm. Code for process $i$}
\label{alg:UnivConstructionGCA}
\end{algorithm}

\paragraph{Overview.}
%
%
%
Algorithm~\ref{alg:UnivConstructionGCA} proceeds in rounds, using a distinct GCA instance $\text{GCA}_r$ at each round $r \geq 1$, together with a shared array $S[1,\ldots,n]$ of SWMR atomic registers.
For a process $i$, entry $S[i]$ stores the latest trace committed by $i$ on the universal construction.
It consists of a pair $(r,s)$, where $r$ is a round number and $s$ is the trace committed by $i$ at that round.

The construction takes, as an input, a sequential object $A=(Q,q_0,O,R,\sigma)$.
Upon invoking an operation $op$, a process $i$ creates a matching \emph{command}, that is, a tuple $(op,i,seq)$ where $seq$ is a local sequence number.
Commands are uniquely identifiable and inherit the conflict relation of their underlying operations; they are used to instantiate GCA.
In detail, $A_c=(Q,q_0,C,R,\sigma')$ is the object derived from $A$, with
$C=\{(op,i,seq) : op\in O,\, i\in \Pi,\, seq\in\mathbb{N}\}$ and
$\sigma'((op,i,seq),q)=\sigma(op,q)$.
The meet-semilattice $(C^*\!/\!\sim,\leq)$ induced by $A_c$ parametrizes every instance $\text{GCA}_k$.
Notice that, in the definition above, the transition function ignores the process identifier and sequence number: executing commands is identical to executing their underlying operations.

To invoke a new operation $op$, process $i$ creates a unique command $cmd$ by combining $op$ with its identifier and a local sequence number (Line~\ref{line:uc:3}).
Then, process $i$ reads the shared array $S$ and takes the pair $(r,s)$ with the largest round number, where $s$ is a committed trace (Line~\ref{line:uc:4}).
Starting from this trace, the process proposes successive extensions that append $cmd$ when absent, always using the trace returned by the previous $\gca$ invocation. 
This loops until a committed trace containing $cmd$ is obtained (Line~\ref{line:uc:5}).
At this point, the process records the pair $(r,s)$ in $S[i]$ and returns the response of $cmd$ in $s$ (Line~\ref{line:uc:11}).

\paragraph{Example.}
To illustrate the above construction, let us go back to Figure~\ref{fig:cf-vs-wcf}(a).
In this figure, processes $p_1$, $p_2$, and $p_3$ invoke operations $read()$, $inc()$, and $dec()$, respectively.
Suppose that they use the universal construction in Algorithm~\ref{alg:UnivConstructionGCA} to implement the shared counter.
A possible execution is that $p_1$ and $p_2$ propose concurrently the traces $u=(read(),1,1)$ and $v=(inc(),2,1)$ to $\text{GCA}_1$.
In that case, since $u$ and $v$ are not compatible, both processes may adopt $v$ (or $\epsilon$).
Now suppose $p_1$ gets stalled whereas $p_3$ invokes its operation.
Then, if $p_3$ is scheduled first, it accesses $\gca_1$, adopts $v$, appends $(dec(),3,1)$ to it, and proposes the resulting trace to $\text{GCA}_2$.
After that, $p_2$ proposes $v$ to $\text{GCA}_2$.
Since the two input traces of $\text{GCA}_2$ are compatible, at least one process commits, observes its command included in the committed trace, and returns from its invocation.

The following theorems establish the correctness of Algorithm~\ref{alg:UnivConstructionGCA}. 
Sketches of their proofs are presented below, with complete details deferred to Appendix~\ref{sec:appendix:weak_uc}.

\begin{restatable}{theorem}{WeakUCLin}
\label{theorem:weakUCLin}
    Algorithm~\ref{alg:UnivConstructionGCA} implements a linearizable universal construction.
\end{restatable}
\begin{proof}[Proof (Sketch)]
The result of an operation is computed when, at some round $r$, its corresponding command $cmd$ is committed, that is, when $\gcapropose{r}{\cdot}$ returns a pair $(t,\True)$, with $cmd \in t$.
Committed traces grow monotonically during the execution: once a command appears in a trace committed in $\gca_r$, it appears in all traces committed in $\gca_{r' \geq r}$.
By Lemma~\ref{lemma:ret-trace-prefix}, executing a command in a given trace yields the same response in every extension of the trace.
Thus, we can linearize the committed operations by taking any representative of the latest committed trace that respects the real-time order of the original history (as we show in the detailed proof, such a representative exists). 
\end{proof}

\begin{restatable}{theorem}{WeakUCWCF}
\label{theorem:weakUCWCF}
    Algorithm~\ref{alg:UnivConstructionGCA} ensures weak conflict-freedom.
\end{restatable}
\begin{proof}[Proof (sketch)]
Consider an infinite execution that is eventually conflict-free.
After some time, all pending commands are pairwise non-conflicting and only correct processes take steps. 
In every later GCA round, processes construct proposals from previous outputs, appending their pending commands; non-conflicting appends preserve
compatibility.
By Convergence, the outputs remain compatible.
Since all participants return, Commitment guarantees that some process commits a trace containing its command and completes its pending operation. 
Hence, it cannot be the case that no operation completes.
For eventual step-contention-freedom, if a process eventually runs solo, it invokes a GCA object as the only participant, necessarily commits, and thus completes its pending call.
\end{proof}

\paragraph{Improving fairness.}
Algorithm~\ref{alg:UnivConstructionGCA} only ensures weak conflict-freedom.
Even if an execution is eventually conflict-free, the algorithm provides no guarantee that \emph{every} invocation by a correct process completes. 
Since GCA offers no fairness, faster processes can repeatedly commit traces that exclude a slower process’s command, preventing it from ever committing. 

We remedy this by augmenting Algorithm~\ref{alg:UnivConstructionGCA} with a straightforward \emph{helping} mechanism: processes publish their commands in a shared announcement array, and extend their proposals with commands from this array, ensuring that every command is eventually included in a committed trace once the execution becomes conflict-free.
This gives a conflict-free read-write universal construction.
A detailed description and a proof of correctness are presented in Section~\ref{sec:univconstructionv2}.

%
%
%

\section{GCA Algorithm} \label{sec:GCAimp}

This section presents a wait-free implementation of the Generalized Commit-Adopt (GCA) abstraction using two atomic snapshot objects~\cite{AADGMS93}.

\paragraph{Overview.}
The construction is described in Algorithm~\ref{alg:GCAFV2}.
It proceeds in two phases using snapshot objects $A$ and $B$.
Object $A$ serves to share the input traces and to compute a compatible combination of them.
If a process observes compatible input traces, it writes a \textit{candidate} output trace in $B$.
Each process outputs the GLB of these candidates, committing it when it gathers enough support by the other processes.
If no such candidate exists, a process adopts its own compatible combination of the input traces.

\begin{algorithm}
\DontPrintSemicolon
\SetAlgoLined
\KwShared{$\textit{A}=(\bot, \ldots, \bot)$ \hfill $\triangleright$ Snapshot object.}
\KwShared{$\textit{B}=((\bot, \bot), \ldots, (\bot, \bot))$ \hfill $\triangleright$ Snapshot object.}
\Fn{propose($s_i$)}{
    $A[i] \gets s_i$ \label{line:gca:1} \\
    $A_i \gets A.\text{snap()}$ $\ \triangleright \ A_i = (a_1, \ldots, a_n)$ \label{line:gca:2} \\
    $A_i^{\mathsf{co}}\!=\!\left( a_1^{\mathsf{co}}, \ldots, a_n^{\mathsf{co}} \right),\; \text{where}\;
    a_k^{\mathsf{co}}\!=\!\bigsqcap(\{a_k\}\cup\{a_j : \neg\textit{comp}({\{a_j, a_k\})}\})$ \label{line:gca:3} \\
    $B[i] \gets (\bigsqcup A_i^{\mathsf{co}},A_i = A_i^\mathsf{co})$ \label{line:gca:4} \\
    $B_i \gets B.\text{snap()}$ $\ \triangleright \ B_i = ((b_1,c_1), \ldots, (b_n,c_n))$ \label{line:gca:5} \\
    $\beta_i =
    \begin{cases}
    \bigsqcap \{ b_k : b_k \neq \bot \land c_k = \True \} & \text{if } \exists j : c_j = \True \\
    b_i & \text{otherwise}
    \end{cases}$ \label{line:gca:6} \\
    $w_i = (\forall a_k: a_k \neq \bot \implies a_k = s_i) \land (\forall b_k: b_k \neq \bot \implies b_k = s_i)$ \label{line:gca:7} \\
    \Return $\Bigl(\beta_i,\ w_i \lor((\forall a_k: a_k \leq \beta_i \implies b_k \neq \bot) \land (\nexists c_k : c_k = \False))\Bigr)$ \label{line:gca:8} \\
}
\caption{Solving Generalized Commit-Adopt. Code for process $i$.}
\label{alg:GCAFV2}
\end{algorithm}

\paragraph{Algorithmic details.}
A process $i$ first shares its input trace $s_i$ by writing it to $A$.
Then, it takes a snapshot of $A$ using variable $A_i$ (Line~\ref{line:gca:2}).
The input traces in $A_i$ may order conflicting commands differently.
As a consequence, $i$ must solve any such disagreement before accessing $B$.
To this end, each entry $a_k \in A_i$ is replaced with the GLB of $a_k$ and any other trace in $A_i$ not compatible with it (Line~\ref{line:gca:3}).%
\footnote{
The definitions of $\bigsqcap$, $\bigsqcup$ and $\comp{}$ extend naturally to admit $\bot$ values.
See Appendix~\ref{sec:appendix:gca} for the details.
}
This yields a set $A_i^{\mathsf{co}}$.
The key observation here is that $A_i^{\mathsf{co}}$ contains only \emph{compatible} traces.
As a result, it admits a well-defined least upper bound (LUB), $\bigsqcup A_i^\mathsf{co}$, which is then stored in $B$.
%
%
Process $i$ also records in $B$ whether the original snapshot $A_i$ is compatible, i.e., whether $A_i = A_i^{\mathsf{co}}$.
This way each entry in $B$ is a pair $(b_i,c_i)$.
When the \emph{compatibility flag} $c_i$ is raised (i.e., set to $\True$), this indicates that the trace $b_i$ is a candidate output trace.

In the second phase, process $i$ takes a snapshot of $B$, stores it in variable $B_i$,
%
%
and uses $B_i$ to compute the output trace, $\beta_i$.
If any of the compatibility flags in $B_i$ is raised, $i$ takes the GLB of the candidates in $B_i$, that is, the pairs of the form $(b_k,\True)$.
Otherwise, process $i$ takes $b_i$. 
%
%
The computation of the output trace $\beta_i$ happens in Line~\ref{line:gca:6}.

To satisfy the Weak-Agreement property of GCA, the processes must commit when all inputs are identical.
To this end, each process has a local \textit{weak-agreement flag}, $w_i$.
This flag is raised when all non-$\bot$ entries in both snapshot views, $A_i$ and $B_i$, are the same as  the input of $i$. 
Alternatively, a process may also commit $\beta_i$ if two conditions hold: (1)~all compatibility flags present in $B_i$ are raised, and (2)~every trace in $A_i$ that is a prefix of $\beta_i$ is supported by a corresponding non-$\bot$ entry in $B_i$. 
In all the other cases, the process adopts $\beta_i$.
This computation is done in Line~\ref{line:gca:8}.

%

\begin{restatable}{theorem}{GCAWorks}
\label{theorem:GCA_works}
    Algorithm~\ref{alg:GCAFV2} is an implementation of Generalized Commit-Adopt.
\end{restatable}

\begin{proof}[Proof (Sketch)]
By snapshot containment, the views obtained from $A$ and $B$ form a chain under inclusion: for any two processes $i,j$, either $ A_i \subseteq A_j $ or $A_j \subseteq A_i$, and similarly for $B$.
When all input traces are compatible, all candidates $b_i$ written to $B$, and hence all traces $\beta_i$, are comparable.
Based on these key observations, let us briefly explain why the properties of GCA hold.
(\emph{Adoption.})
If $i$ commits $\beta_i$, for any input $a_j$ such that $a_j \not\leq \beta_i$, process $j$ observes a larger view in $A$ and writes a candidate $b_j \geq \beta_i$.
Now if $a_j \leq \beta_i$, then $b_j \neq \bot$ in $B_i$, and thus $b_j \geq \beta_i$.
Last, for any input $a_j=\bot$ in $A_i$, by containment $b_j \geq b_i \geq\beta_i$.
From what precedes, for any process $j$, $\beta_j \geq \beta_i$.
(\emph{Commitment.})
Suppose all inputs are compatible.
As all output traces are comparable, there exists a smallest output trace $\beta$.
Because the views of $B$ form a containment chain, some process $j$ such that $b_j=\beta$ observes all entries that justify $\beta$ and thus commits an extension of its input.
(\emph{Convergence.})
Since for all~$i$, $\beta_i=\bigsqcap B_i$ or $\beta_i=b_i$, by containment on $B$, all output traces are compatible.
(\emph{Validity}, \emph{Common-Prefix}.)
These two properties hold because every $\beta_i$ is obtained solely via GLB/LUB combinations of the input traces.
(\emph{Weak Agreement}.)
The flag $w_i$ guarantees that this property holds when all the observed traces coincide.
We refer the reader to Appendix~\ref{sec:appendix:gca} for a detailed proof.
\end{proof}

\section{Helping: Conflict-Free Universal Construction} 
\label{sec:univconstructionv2}

Algorithm~\ref{alg:UnivConstructionGCAV2} augments Algorithm~\ref{alg:UnivConstructionGCA} with a standard helping mechanism: processes announce their commands, so that other processes can help them make progress, ensuring that these commands complete in eventually conflict-free executions.

In detail, upon invoking a new operation, process $i$ \emph{publishes} it by writing the corresponding command to the \emph{announcement array} $M[1,\ldots,n]$ (Line~\ref{line:cfuc:5}).
%
%
Process $i$ then reads the shared array $S$ and retrieves the highest round $r^*$ at which a trace $s^*$ was committed.
If $r^*>r$, it updates its local variables to $(r,s)\leftarrow(r^*,s^*)$.
As in Algorithm~\ref{alg:UnivConstructionGCA}, this trace supersedes traces committed at prior rounds because it suffixes them.
The comparison $r^*>r$ is necessary since process $i$ may have completed a previous invocation after reading, at Line~\ref{line:cfuc:16}, a committed trace from a round lower than the round it had reached locally.

At each round, process $i$ constructs a proposal by extending the longest known committed trace with commands in $M$ missing from it.
To this end, process $i$ first computes the commands in $M$ missing from the committed trace (Line~\ref{line:cfuc:10}).
Then, it uses the operator $\mathit{trace}: 2^\calC \rightarrow \calC^*\!/\!\sim$ to create a trace from them by arranging the commands in some arbitrary order, and then appends the resulting trace to the committed trace (Line~\ref{line:cfuc:11}).
Further, process $i$ accesses the GCA instance of the current round and proposes the new trace.
If process $i$ commits at some round, it updates its entry in the shared array $S$.

The previous steps are executed in a loop, until a trace containing the command corresponding to the current call is committed (Lines~\ref{line:cfuc:8}~and~\ref{line:cfuc:16}).
At this point, the process computes an appropriate response, as in Algorithm~\ref{alg:UnivConstructionGCA}.

\begin{algorithm}
\DontPrintSemicolon
\SetAlgoLined
\KwInput{$A= (Q,q_0,O,R,\sigma)$ \hfill $\triangleright$ Sequential object specification.}
\KwShared{$\{\gca_k(C^*\!/\!\sim,\leq): k \geq 1\}$ \hfill $\triangleright$ Sequence of GCA objects.}
\KwShared{$S[1,\ldots, n] = (0,\epsilon);\; M[1,\ldots,n] = \bot$ \hfill $\triangleright$ Arrays of atomic registers.}
\KwInitial{$r = 0;\;seq=0;\;s=\epsilon$}
\Fn{invoke(op)}{
    \Let $c = \False$ \\
    $seq \gets seq + 1$ \\
    $cmd \gets (op,i,seq)$ \label{line:cfuc:4}\\
    $M[i] \gets cmd$ \label{line:cfuc:5}\\

    $(r^*,s^*) \gets \max_{r'} \{(r',s')=S[j] : j \in [1,n]\};\; \textbf{if } r^*\!>\!r \textbf{ then } (r,s)\gets(r^*,s^*)$ \label{line:cfuc:6}

    \Let $u = \epsilon$ \\
    
    \While{$cmd \notin \ops{u}$ \label{line:cfuc:8}}{
        $r \gets r+1$ \label{line:cfuc:inc}\\
        $M_i \gets \{M[j] : j \in [1,n] \land M[j] \neq \bot \land M[j] \notin \ops s\}$ \label{line:cfuc:10}\\
        $s \gets s \cdot \mathit{trace}(M_i)$ \label{line:cfuc:11} \\
        $(s,c) \gets \text{GCA}_r.\text{propose}(s) $ \label{line:cfuc:12}\\
        \If{$c = \True$}{
            $S[i] \gets (r,s)$ \label{line:cfuc:14}\\
        }
        $(\_,u) \gets \max_{r'} \{(r',s')=S[j] : j \in [1,n]\}$ \label{line:cfuc:16}\\
    } \label{line:cfuc:endwhile}
    \Return\ $\ret^*(cmd,u)$ \label{line:cfuc:19}\\
}
\caption{Conflict-free universal construction algorithm. Code for process $i$}
\label{alg:UnivConstructionGCAV2}
\end{algorithm}

\subsection{Correctness of the Conflict-free Universal Construction}

\begin{lemma}
Algorithm~\ref{alg:UnivConstructionGCAV2} implements a linearizable universal construction.
\end{lemma}
\begin{proof}
The proof is in the same vein as that of Theorem~\ref{theorem:weakUCLin}.
Consider some execution $\alpha$ and let $H$ be the corresponding history.
Let $r^*$ be the largest round in which a trace $t$ is committed in $\alpha$.
According to Line~\ref{line:cfuc:4}, each invocation of an
operation $\Phi$ creates a unique command, $cmd(\Phi)$.
Operation $\Phi$ completes only after reading, at Line~\ref{line:cfuc:16}, the trace $u$ with the
largest round number in $S$ such that $cmd(\Phi) \in \ops u$.
The response of $\Phi$ in $H$, computed in Line~\ref{line:cfuc:19}, is $\ret^*(cmd(\Phi),u)$.
Because $u \leq t$, this is also the response of $\Phi$ computed from $t$, that is $\ret^*(cmd(\Phi),t)$.
A sequential history is built from a representative schedule $\hat t \in t$ that respects the real-time order of invocations, exactly as in the proof of Theorem~\ref{theorem:weakUCLin}.
This sequential history is a linearization of the history~$H$. 
\end{proof}

\hypertarget{lemma62}{}
\begin{lemma} \label{lemma:UCV2isCF}
Algorithm~\ref{alg:UnivConstructionGCAV2} ensures conflict-freedom.
\end{lemma}
\begin{proof}
Let $\alpha$ be some infinite execution that is eventually conflict-free.
By definition, there exists a point in time $\tau$ after which no two pending operations are conflicting.
%
%
For the sake of contradiction, suppose there exists an operation instance $\Phi$ invoked by a correct process $i$ which remains pending forever in $\alpha$.
Let $cmd(\Phi)$ be the command created by process $i$ in Line~\ref{line:cfuc:4}.
From what precedes, there exists a time $\tau' \geq \tau$ after which only correct processes take steps, and process $i$ executed Line~\ref{line:cfuc:5}, writing $cmd(\Phi)$ in $M[i]$.
Let $r^*$ be the largest round reached by any process by time $\tau'$.

Assume that the following invariants hold:
(I) for any $r>r^*$ only compatible traces are proposed to $\text{GCA}_{r}$.
(II) for any $r>r^*$ every output trace of $\gca_{r}$ contains $cmd(\Phi)$. 
%
%
By invariant (I), all proposals to $\gca_{r^*+1}$ are compatible.
By Commitment, there exists a correct process $j$ that invokes $\gcapropose{r^*+1}{s_j}$, and returns $(x, \True)$ such that $s_j \leq x$.
Let $t_j$ be the trace retrieved by process $j$ after invoking $\gcapropose{r^*}{\cdot}$.
Then by Line~\ref{line:cfuc:11}, $s_j=t_j \cdot m_j$, where $m_j= \mathit{trace}(M_j)$.
Since we are at some time after $\tau'$, $M[i]=cmd(\Phi)$.
If $cmd(\Phi) \notin t_j$, then $cmd(\Phi) \in \ops{\mathit{trace}(M_j)}$.
Otherwise, $cmd(\Phi) \in t_j$.
In any case, $cmd(\Phi) \in \ops {s_j}$.
Therefore, process $j$ commits $x$ and $cmd(\Phi) \in \ops x$.
Then, process $j$ writes $(x,r^*+1)$ in $S$ at Line~\ref{line:cfuc:14}.
As a result, at round $r^*+1$ process $i$ will read $S$ at Line~\ref{line:cfuc:16} and $x=u$. 
Thus, $cmd(\Phi) \in u$ and process $i$ exits the while loop at Lines~\ref{line:cfuc:8}-\ref{line:cfuc:endwhile}, completing its operation---a contradiction.

In what follows, we prove that invariant (I) holds.
Consider some round $r > r^*$.
Let $j$ be some (correct) process which participates in $\gca_r$ and $cmd_j$ be the command associated with the operation invoked by $j$.
Let $t_j$ and $s_j$ be the traces $j$ retrieves from $\gca_{r-1}$ and proposes to $\gca_{r}$, respectively.
According to Line~\ref{line:cfuc:11}, $s_j= t_j \cdot m_j$, where $m_j=trace(M_j)$.
By definition of the ${trace}()$ operator, $\ops {m_j} = M_j$.
We establish that $\mathcal{S} = \bigcup_i \{s_i\}$ is compatible.

By the Convergence property, $\mathcal{T}= \bigcup_i \{t_i\}$ is compatible.
Let $t$ be the common prefix of these traces, that is $t = \bigsqcap \mathcal{T}$.
Given a correct process $i$, $u_i$ is the (unique) trace such that $t_i = t \cdot u_i$.
%
%
Let $cmd(\Phi')$ be a command of an operation instance $\Phi'$ which has completed in the execution.
Since $\Phi'$ completed, there exists a round $r'$ where some process output $(y,\True)$ from $\gca_{r'}$ and $cmd(\Phi') \in y$.
We have that $r' \leq r^*$.
If $r'>r^*$ by invariant (II), $cmd(\Phi) \in \ops{y}$ and since $\Phi'$ completes, the process that commits $y$ writes it in $S$ at Line~\ref{line:cfuc:14}, requiring $\Phi$ to complete --- a contradiction.
Therefore, $r' \leq r^* < r$ and by applying Lemma~\hyperlink{lemmaA1}{\ref*{lemma:prefix-rounds}}, $y \leq t$, implying $cmd(\Phi') \in \ops{t}$.
Thus, $u_i$ contains only commands from pending operations.
Moreover, according to Line~\ref{line:cfuc:10}, this is also the case of $m_i$.
Thus, the set $\bigcup_j \ops{u_j} \cup \bigcup_j \ops{m_j}$ only contains pending operations.
Since we are at a time $\tau' > \tau$, these operations do not conflict between each other.
It follows that $\mathcal{S} = \bigcup_i \{t \cdot u_i \cdot   m_i\}$ only contains compatible traces, as required.

We now prove invariant (II).
Consider some round $r > r^*$.
Let $j$ be some (correct) process which participates in $\gca_r$.
Let $t_j$ and $s_j$ be the traces $j$ retrieves from $\gca_{r-1}$ and proposes to $\gca_{r}$, respectively.
According to Line~\ref{line:cfuc:11}, $s_j= t_j \ccdot m_j$, where $m_j=trace(M_j)$.
Define $t = \bigsqcap \bigcup_k \{t_k\}$, where $t_j = t \ccdot u_j$.
Then $s_j = t \ccdot u_j \ccdot m_j$.
Since we are at some time after $\tau'$, $M[i]=cmd(\Phi)$ at Line~\ref{line:cfuc:10}, thus $cmd(\Phi) \in \ops{s_j}$.
Let $cmd(\Phi') \in \ops{u_j \ccdot m_j}$ be some command of an operation instance $\Phi'$.
If $\Phi'$ completed before $\tau'$, $cmd(\Phi')$ is in some trace committed in a round smaller than $r^*+1$.  
By Lemma~\hyperlink{lemmaA1}{\ref*{lemma:prefix-rounds}}, $cmd(\Phi') \in \ops{t}$ and $cmd(\Phi') \notin\ops{u_j \ccdot m_j} $.
Thus, $\Phi'$ is pending after time $\tau'$ which implies that $cmd(\Phi') \not\asymp cmd(\Phi)$.
Therefore, $cmd(\Phi)$ does not conflict with any command in $u_j \ccdot m_j$ and thus 
$cmd(\Phi) \in \ops{\bigsqcap\bigcup_k\{s_k\}}$.
By GCA Common Prefix, $cmd(\Phi)$ is in every output trace of $GCA_r$.

It remains to show that for any infinite execution $\beta$ and any operation instance $\Phi$ of a correct process $i$ that is eventually step-contention-free in $\beta$, $\Phi$ completes.
This part of the proof is identical to the one provided for Theorem~\ref{theorem:weakUCWCF}.
In detail, there exists a time $\tau$ such that after $\tau$ only $i$ takes steps in $\beta$.
Let $r^*$ be the largest round reached by any process by time $\tau$.
Since $i$ continues executing the loop in Lines~\ref{line:cfuc:8}-\ref{line:cfuc:endwhile}, it eventually invokes $\gcapropose{r^*+1}{s}$ with $cmd(\Phi)\in \ops{s}$.
As $i$ is the only participant to $\gca_{r^*+1}$, the call returns $(s,\True)$, with $cmd(\Phi) \in \ops s$ and process $i$ writes $(s,r^*+1)$ in $S[i]$ (Line~\ref{line:cfuc:14}). 
Then, as $i$ is the only process to reach round $r^*+1$, it assigns $u=s$ after reading $S$ at Line~\ref{line:cfuc:16}.
Thus, the process exits the while loop and $\Phi$ completes.
\end{proof}

\section{Resolving conflicts in step-contention-free runs} \label{sec:resolution}

Strictly speaking, in a (weakly) conflict-free implementation, an incomplete operation invoked by a faulty process may affect progress of \emph{every} future operation.
It may therefore appear that a single failure can bring our progress criteria down to plain obstruction-freedom. 

However, a closer look at our conflict-free algorithm reveals that the conflicts inflicted by faulty processes can be made ``harmless''. 
In our conflict-free universal construction (Algorithm~\ref{alg:UnivConstructionGCAV2}), the helping mechanism tries to take care of all pending operations, including those invoked by faulty processes. 
Therefore, it may appear that the conflicts incurred by faulty processes may stay in the system forever. 
However, all conflicts are resolved if some process manages to commit its trace \emph{in a step-contention-free execution}.
Indeed, a process running solo commits a trace containing every pending command, and every trace returned in a later GCA instance extends it (Lemma~\hyperlink{lemmaA1}{\ref*{lemma:prefix-rounds}}).
%
%
%
From this point on, ``past'' conflicts caused by the processes that have already failed are not going to hinder progress. 
We can therefore strengthen our liveness guarantee by focusing only on conflicts created by operations which are invoked \emph{after} a process manages to access solo some GCA instance.

Formally, let $A$ be a conflict-free algorithm, let $\alpha$ be a finite execution of $A$ and let $\alpha'$ be an infinite execution of $A$ extending $\alpha$. 
We say that $\alpha'$ is \emph{eventually $\alpha$-conflict-free} if it has a suffix in which no two concurrent operation instances $\Phi=(o,i)$ and $\Phi'=(o',j)$ \textbf{invoked after $\alpha$} satisfy $o \asymp o'$.

\begin{definition}[Conflict-resolving execution]\label{def:wcr}
A finite execution $\alpha$ of $A$ is \emph{conflict-resolving} if in every eventually $\alpha$-conflict-free infinite extension of $\alpha$, every correct process completes each of its operations.
\end{definition}

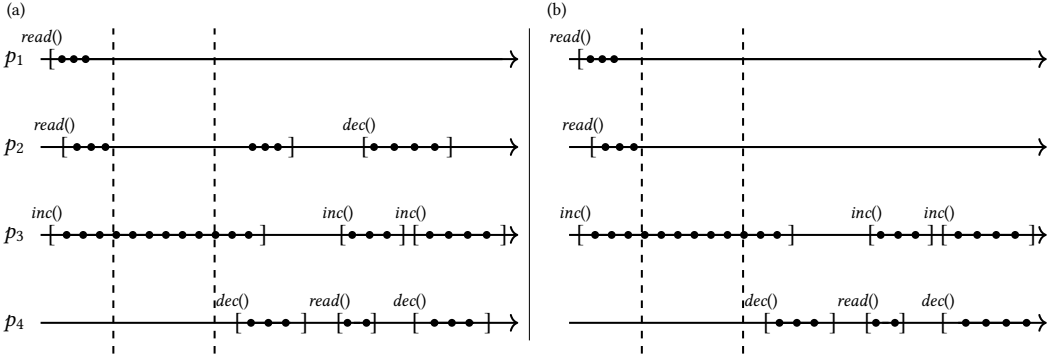
\begin{figure}
\centering
\resizebox{\columnwidth}{!}{%
\begin{tikzpicture}[x=0.88cm,y=1.22cm,
  proc/.style={font=\small,anchor=east},
  timeline/.style={thick,->},
  br/.style={font=\small,inner sep=0pt},
  lab/.style={font=\scriptsize,inner sep=0pt}
]

\def\xStart{0.05}
\def\xEnd{7.60}
\def\panelShift{7.35}

\pgfmathsetmacro{\dotRatio}{0.88/1.22}
\def\dotR{0.05}
\def\dotLW{0.6pt}

\newcommand{\op}[4]{%
  \node[br] at (#2,#1) {$[\,$};
  \node[lab,anchor=east] at (#2+0.155,#1+0.24) {#4};
  \ifx&#3&%
    \draw[thick] (#2+0.20,#1) -- (\xEnd-0.25,#1);
  \else
    \draw[thick] (#2+0.20,#1) -- (#3-0.20,#1);
    \node[br] at (#3,#1) {$\,]$};
  \fi
}

\newcommand{\uniformsteps}[4]{%
\pgfmathsetmacro{\step}{(#3-#2)/(#4+1)}
\foreach \k in {1,...,#4} {%
  \pgfmathsetmacro{\xx}{#2+\k*\step}
  \filldraw[line width=\dotLW] (\xx,#1) ellipse
    [x radius=\dotR, y radius=\dotR*\dotRatio];
}%
}

\def\lblX{-0.05}

\node[proc] at (\lblX,3) {$p_1$};
\node[proc] at (\lblX,2) {$p_2$};
\node[proc] at (\lblX,1) {$p_3$};
\node[proc] at (\lblX,0) {$p_4$};

\pgfmathsetmacro{\midXBase}{(\xEnd + (\panelShift+\xStart))/2.0}
\def\sepNudge{0.275}
\pgfmathsetmacro{\midX}{\midXBase+\sepNudge}
\draw[thin] (\midX,-0.2) -- (\midX,3.35);

\node[font=\scriptsize,anchor=east] at (\lblX,3.55) {(a)};
\node[font=\scriptsize,anchor=west] at (\midX+0.14,3.55) {(b)};

\begin{scope}
  \draw[timeline] (\xStart,3) -- (\xEnd,3);
  \draw[timeline] (\xStart,2) -- (\xEnd,2);
  \draw[timeline] (\xStart,1) -- (\xEnd,1);
  \draw[timeline] (\xStart,0) -- (\xEnd,0);

  \op{3}{0.25}{}{\textit{read}()}

  \op{2}{0.45}{4}{\textit{read}()}
  \op{2}{5.2}{6.5}{\textit{dec}()}

  \op{1}{0.25}{3.55}{\textit{inc}()}
  \op{1}{4.85}{5.75}{\textit{inc}()}
  \op{1}{6}{7.35}{\textit{inc}()}

  \op{0}{3.2}{4.2}{\textit{dec}()}
  \op{0}{4.8}{5.3}{\textit{read}()}
  \op{0}{6}{7.1}{\textit{dec}()}
  
  \draw[dashed,thick] (1.2,-0.35) -- (1.2,3.35);
  \draw[dashed,thick] (2.8,-0.35) -- (2.8,3.35);

  \uniformsteps{3}{0.2}{0.95}{3}

  \uniformsteps{2}{0.4}{1.3}{3}
  \uniformsteps{2}{3.2}{4}{3}  
  \uniformsteps{2}{5}{6.6}{4}  
  \uniformsteps{1}{0.2}{3.6}{12}
  \uniformsteps{1}{4.7}{5.8}{3}
  \uniformsteps{1}{5.9}{7.4}{4}
  \uniformsteps{0}{3.1}{4.2}{3}
  \uniformsteps{0}{4.6}{5.5}{2}  
  \uniformsteps{0}{6}{7.1}{3}
  
\end{scope}

\begin{scope}[xshift=\panelShift cm]
  \draw[timeline] (\xStart,3) -- (\xEnd,3);
  \draw[timeline] (\xStart,2) -- (\xEnd,2);
  \draw[timeline] (\xStart,1) -- (\xEnd,1);
  \draw[timeline] (\xStart,0) -- (\xEnd,0);
  
  \op{3}{0.25}{}{\textit{read}()}

  \op{2}{0.45}{}{\textit{read}()}

  \op{1}{0.25}{3.55}{\textit{inc}()}
  \op{1}{4.85}{5.75}{\textit{inc}()}
  \op{1}{6}{7.35}{\textit{inc}()}

  \op{0}{3.2}{4.2}{\textit{dec}()}
  \op{0}{4.8}{5.3}{\textit{read}()}
  \op{0}{6}{}{\textit{dec}()}
  
  \draw[dashed,thick] (1.2,-0.35) -- (1.2,3.35);
  \draw[dashed,thick] (2.8,-0.35) -- (2.8,3.35);

  \uniformsteps{3}{0.2}{0.95}{3}

  \uniformsteps{2}{0.4}{1.3}{3}
  \uniformsteps{1}{0.2}{3.6}{12}
  \uniformsteps{1}{4.7}{5.8}{3}
  \uniformsteps{1}{5.9}{7.4}{4}
  \uniformsteps{0}{3.1}{4.2}{3}
  \uniformsteps{0}{4.6}{5.5}{2}  
  \uniformsteps{0}{6}{7.6}{4}
\end{scope}

\end{tikzpicture}%
}
\caption{
Examples of (a) conflict-resolving and (b) conflict-forgetting executions on a shared counter.
In both executions, $p_1$ crashes and $p_3$ runs solo in the interval between the dashed lines.
%
%
In execution (a), all correct processes ($p_2,p_3,p_4$) complete their operations because process $p_3$ resolved the conflicts during its contiguous solo steps.
In (b), these steps permit to forget the conflicts and make progress as long as $p_2$ does not wake up.
}
\Description{
Executions which are (a) conflict-resolving and (b) conflict-forgetting.
Execution (a): $p_1$ invokes a read and then crashes; $p_2$ invokes a read, and temporally stops taking steps; $p_3$ invokes an increment, and at some point takes sufficiently many steps in the absence of step-contention.
The read operations of $p_1$ and $p_2$, and the increment of $p_3$ are concurrent. 
After taking sufficiently many step-contention-free steps, $p_3$ completes its increment operation.
Suppose that the execution becomes conflict-resolving.
$p_2$ resumes taking steps and completes its read operation.
Then, at some point processes $p_2$, $p_3$ and $p_4$ only invoke decrements and increments.
Then, every correct process completes each of its operations.
Execution (b): The same as execution (a) but process $p_2$ crashes after taking some steps of its read invocation.
}
\label{fig:resolving}
\end{figure}

As an example, consider execution (a) in Figure~\ref{fig:resolving}, in which four processes invoke operations on a shared counter.
Process $p_1$ invokes a read operation that never completes, while process $p_2$ invokes another read operation and temporarily stops taking steps.
Process $p_3$ then invokes an increment operation and executes sufficiently many steps without step contention.
At some point, any two concurrent operations invoked after the solo execution of $p_3$ are non-conflicting.
If the solo execution of $p_3$ is conflict-resolving, then every correct process eventually completes its operation.
In particular, after the solo execution of $p_3$, the pending read operation of $p_2$ completes after finitely many additional steps.
Thus, every conflict that arose before the solo execution of $p_3$ is resolved.

We show that our conflict-free universal construction allows for resolving any conflict by running in the absence of step contention. 
Specifically, if a process runs long enough in the absence of step contention, it commits a trace containing every pending command of a pending operation instance.

\hypertarget{theorem64}{}
\begin{theorem}[Conflict resolution]\label{th:cr}
For every finite execution $\alpha$ of Algorithm~\ref{alg:UnivConstructionGCAV2} and every process $i$, there is a finite conflict-resolving $i$-solo extension of $\alpha$.
\end{theorem}
\begin{proof}
Let $\alpha$ be a finite execution of Algorithm~\ref{alg:UnivConstructionGCAV2}. 
Let $r_0$ be the largest round reached by any process in $\alpha$, and let $M_0$ be the state of the shared array $M$ at the end of $\alpha$.
In the induced history of $\alpha$, we consider an operation instance as invoked once its associated command is written into $M$ (Line~\ref{line:cfuc:5}).
Take any process $i$ and let it run after $\alpha$ in the absence of step contention, invoking a new operation instance whenever its current one completes, until it proposes some trace $s'$ to $\gca_{r_1}$ for some round $r_1 > r_0$.
Since $i$ executes solo after $\alpha$, every entry $j \neq i$ of $M$ remains equal to $M_0[j]$, and $M[i]$ holds the command of the current operation instance of $i$.
By Lines~\ref{line:cfuc:10}-\ref{line:cfuc:11}, $s'$ thus contains every command present in $M$.
Since no other process accesses $\gca_{r_1}$ before $i$ returns from it, by GCA Commitment and Validity, $i$ returns $(s',\True)$ and writes $(r_1,s')$ in $S$ at Line~\ref{line:cfuc:14}.
It follows that $r_1$ is the highest round in $S$. 
Process $i$ then reads $s'$ at Line~\ref{line:cfuc:16} and completes its current operation instance.
Let $\alpha'$ denote this extension of $\alpha$.

We now show that $\alpha'$ is conflict-resolving.
Every operation instance invoked before the end of $\alpha'$ has its associated command in $\ops{s'}$: if it completed by the end of $\alpha'$, its command lies in a trace committed at a round at most $r_1$, which prefixes $s'$ by Lemma~\hyperlink{lemmaA1}{\ref*{lemma:prefix-rounds}}; otherwise, its command is in $M$, and therefore in $s'$.
Since $(r_1,s')$ is written in $S$, every such operation instance of a correct process completes at Line~\ref{line:cfuc:16}.
Moreover, by Lemma~\hyperlink{lemmaA1}{\ref*{lemma:prefix-rounds}}, $s'$ prefixes every trace returned by $\gca_r$ for $r \geq r_1$.

Consider any infinite extension $\alpha''$ of $\alpha'$ that is eventually $\alpha'$-conflict-free.
By definition, there exists a time $\tau$ after which no two concurrent operation instances invoked after $\alpha'$ are conflicting.
For the sake of contradiction, suppose that some operation instance $\Phi$ of a correct process $j$ remains pending forever in $\alpha''$; from what precedes, $\Phi$ is invoked after $\alpha'$.
Let $\tau' \geq \tau$ be a time, later than $\alpha'$, after which only correct processes take steps and such that $j$ has executed Line~\ref{line:cfuc:5}, writing $cmd(\Phi)$ in $M[j]$.
Let $r_2 \geq r_1$ be the largest round reached by any process by time $\tau'$.
Suppose the following invariants hold:
(I) for any $r > r_2$, only compatible traces are proposed to $\gca_r$;
(II) for any $r > r_2$, every trace returned by $\gca_r$ contains $cmd(\Phi)$.
By Invariant (I), all proposals to $\gca_{r_2+1}$ are compatible.
As $\gca_{r_2+1}$ is accessed only after $\tau'$, all its participants are correct and return, and $j$ is one of them since it never leaves its loop.
By GCA Commitment, some process thus commits a trace $x$ at round $r_2+1$ and writes $(r_2+1,x)$ in $S$ at Line~\ref{line:cfuc:14}.
Every trace that $j$ subsequently reads from $S$ at Line~\ref{line:cfuc:16} is committed at a round greater than $r_2$, hence contains $cmd(\Phi)$ by Invariant (II), and $j$ completes the invocation of $\Phi$ --- a contradiction.

It remains to prove the two invariants.
Consider some round $r > r_2$.
Given a correct process $k$ proposing to $\gca_r$, let $t_k$ and $s_k$ be the traces $k$ retrieves from $\gca_{r-1}$ and proposes to $\gca_r$, respectively.
According to Line~\ref{line:cfuc:11}, $s_k = t_k \cdot m_k$, where $m_k = \mathit{trace}(M_k)$ and $\ops{m_k} = M_k$.
Let $t = \bigsqcap \bigcup_k \{t_k\}$ and let $u_k$ be the trace such that $t_k = t \cdot u_k$.
No command of $u_k \cdot m_k$ is in $\ops t$: for $u_k$ by definition, and for $m_k$ by the filter of Line~\ref{line:cfuc:10}.
As $r-1 \geq r_1$, we have $s' \leq t$, so that every command of $u_k \cdot m_k$ belongs to an operation instance invoked after $\alpha'$.

We now prove Invariant (I).
Let $cmd(\Phi') \in \ops{u_k \cdot m_k}$.
Suppose that $\Phi'$ completes, where an output trace $y$ of $\gca_{\hat r}$ is committed with $cmd(\Phi') \in \ops y$.
If $\hat r < r$, Lemma~\hyperlink{lemmaA1}{\ref*{lemma:prefix-rounds}} gives $y \leq t$, so that $cmd(\Phi') \in \ops t$, which is excluded.
Otherwise $\hat r > r_2$, so that $cmd(\Phi) \in \ops y$ by Invariant (II), and the write of $y$ at Line~\ref{line:cfuc:14} lets $j$ complete $\Phi$ at Line~\ref{line:cfuc:16}---a contradiction.
Hence every command of $\bigcup_k \ops{u_k \cdot m_k}$ belongs to an operation instance that is invoked after $\alpha'$ and pending forever, and since we are after time $\tau$, they do not conflict.
Thus, the set $\bigcup_k \{t \cdot u_k \cdot m_k\}$ is compatible, as required.

It remains to prove Invariant (II).
Since we are at some time after $\tau'$, $M[j] = cmd(\Phi)$, so that $cmd(\Phi) \in \ops{s_k}$ by Lines~\ref{line:cfuc:10}-\ref{line:cfuc:11}.
Let $cmd(\Phi') \in \ops{u_k \cdot m_k}$.
Instance $\Phi'$ has not completed before $\tau'$: otherwise $cmd(\Phi')$ would be included in a trace committed at a round at most $r_2$, hence in $\ops t$ by Lemma~\hyperlink{lemmaA1}{\ref*{lemma:prefix-rounds}}.
Thus $\Phi'$ is pending at some time after $\tau'$, at which $\Phi$ is pending as well; as both are invoked after $\alpha'$, $cmd(\Phi')$ and $cmd(\Phi)$ do not conflict.
Hence $cmd(\Phi)$ commutes with every other command of $u_k \cdot m_k$, so that $cmd(\Phi) \in \ops{\bigsqcap_k s_k}$, whether or not $cmd(\Phi) \in \ops t$.
By GCA Common Prefix, every trace returned by $\gca_r$ contains $cmd(\Phi)$.
\end{proof}

Therefore, we can overcome failures and guarantee that \emph{every} correct process terminates by letting some process to resolve conflicts in a solo execution. 
In Algorithm~\ref{alg:UnivConstructionGCAV2}, this is achieved through the helping mechanism: a process running solo commits a trace containing every command currently present in the shared array $M$.

The weakly conflict-free universal construction in Algorithm~\ref{alg:UnivConstructionGCA} does not include such a mechanism.
Nevertheless, once a process commits a trace in a solo execution, pending conflicting operations can hinder progress only if their invoker processes take further steps and reintroduce them into subsequent proposals.
Conflicts incurred by faulty processes are ``forgotten''.
Theorem~\ref{th:WeakUCresolve} formalizes this intuition. 

Let $A$ be an implementation of an object with conflict relation $\asymp$, let $\alpha$ be a finite execution of $A$ and let $\alpha'$ be an infinite execution of $A$ extending $\alpha$. 
We say that $\alpha'$ is \emph{eventually weakly $\alpha$-conflict-free} if it has a suffix in which no two concurrent operation instances $\Phi=(o,i)$ and $\Phi'=(o',j)$ \textbf{that take steps after $\alpha$} satisfy $o \asymp o'$.

\begin{definition}[conflict-forgetting execution]
    A finite execution $\alpha$ of $A$ is \emph{conflict-forgetting} if in every eventually weakly $\alpha$-conflict-free infinite extension of $\alpha$, some correct process that takes infinitely many steps completes each of its operations.
\end{definition}

Consider execution (b) in Figure~\ref{fig:resolving}.
In contrast to execution (a), process $p_2$ also crashes.
Suppose that the solo execution of $p_3$ is conflict-forgetting.
Since the read invocations of $p_1$ and $p_2$ take no further steps, and no conflicting operations invoked by $p_3$ and $p_4$ remain pending, process $p_3$ eventually completes all of its invocations.
Observe that this progress guarantee no longer holds if either $p_1$ or $p_2$ takes a step after the solo execution of $p_3$.
In this case, the conflicting read operations of $p_1$ and $p_2$ are forgotten, but not resolved.

For conflicts to be forgotten, no later participant of the solo GCA instance may reintroduce a command that the solo trace omits.
The GCA specification of Section~\ref{sec:gca} does not guarantee this: by Adoption, a later participant may return a strict extension of the solo trace.
We therefore rely on one additional property, stated in the notation of Section~\ref{sec:gca}.

\begin{itemize}
    \item \textbf{Solo Agreement.} If a process returns before any other process participates, every output trace equals its own.
$P' = P'_r = \{j\} \implies \forall i \in P_r,~ t_i = t_j$, where $P'$ and $P'_r$ are the sets of participants and returning participants in some prefix of the execution.
\end{itemize}       

Algorithm~\ref{alg:GCAFV2} satisfies Solo Agreement (Lemma~\hyperlink{lemmaA8}{\ref*{lemma:GCA_soloagg}} in Appendix~\ref{sec:appendix:gca}); in the rest of this section, the GCA objects used by Algorithm~\ref{alg:UnivConstructionGCA} are assumed to satisfy it.

\begin{theorem}[Conflict forgetting]\label{th:WeakUCresolve}
    For every finite execution $\alpha$ of Algorithm~\ref{alg:UnivConstructionGCA} and every process $i$, there is a finite conflict-forgetting $i$-solo extension of $\alpha$.
\end{theorem}
\begin{proof}
    Let $\alpha$ be a finite execution of Algorithm~\ref{alg:UnivConstructionGCA}.
    We extend $\alpha$ by letting process $i$ run in the absence of step contention until it invokes $\gcapropose{r^*}{s}$ for some round $r^*$ greater than every round reached in $\alpha$, returns from it, and completes its operation instance.
    Let $\alpha'$ denote this extension.
    Since $i$ is the sole participant of $\gca_{r^*}$ in $\alpha'$, by GCA Commitment and Validity it commits $s$ and, by Solo Agreement, every process that returns from $\gca_{r^*}$ in any extension of $\alpha'$ returns $s$.

    We now show that $\alpha'$ is conflict-forgetting.
    Consider any infinite extension $\alpha''$ of $\alpha'$ that is eventually weakly $\alpha'$-conflict-free, and suppose, for the sake of contradiction, that no correct process completes each of its operations in $\alpha''$.
    Since $\alpha''$ is eventually weakly $\alpha'$-conflict-free, there exists a time $\tau$ after which no two pending operation instances that take steps after $\alpha'$ are conflicting. 
    Moreover, there exists a time $\tau' \geq \tau$, later than $\alpha'$, after which only correct processes take steps and all their operations are pending forever. 
    Let $r_0$ be the highest round number reached by some process at time $\tau'$.
    We prove the same invariant as in Theorem~\ref{theorem:weakUCWCF}: (I) for any $r >r_0$, only compatible traces are proposed to $\gca_r$.
    By the same argument as in the proof of Theorem~\ref{theorem:weakUCWCF}, this implies that some process completes its operation after time $\tau'$, a contradiction.

    It remains to prove invariant (I). 
    Consider some round $r > r_0$.
    Given a correct process $j$ proposing to $\gca_r$, let $cmd_j$ be the command of its forever-pending operation instance, and let $t_j$ and $s_j$ be the traces $j$ retrieves from $\gca_{r-1}$ and proposes to $\gca_r$, respectively.
    Let $t = \bigsqcap \bigcup_j \{t_j\}$ and let $u_j$ be the trace such that $t_j = t \cdot u_j$; by Lines~\ref{line:uc:7}-\ref{line:uc:8}, $s_j = t \cdot u_j \cdot cmd_j$ if $cmd_j \notin \ops{t_j}$, and $s_j = t \cdot u_j$ otherwise.
    As no operation completes after $\tau'$, every command of a completed operation is in a trace committed at a round at most $r_0$, which prefixes $t$ by Lemma~\hyperlink{lemmaA1}{\ref*{lemma:prefix-rounds}}.
    Thus, every command of $\bigcup_j \ops{u_j} \cup \bigcup_j \{cmd_j\}$ belongs to an operation instance pending forever.
    It remains to show that each such instance takes a step after $\alpha'$.
    This is clear for the instance of each $cmd_j$, since $j$ proposes to $\gca_r$ after $\tau'$.
    We show, by induction on $\hat r \geq r^*$, that every command $cmd'$ of a trace returned by $\gca_{\hat r}$ that is not in $\ops s$ is appended at Line~\ref{line:uc:8} by its own process when proposing to some $\gca_{r'}$ with $r^* < r' \leq \hat r$.
    For $\hat r = r^*$ there is no such command, since every trace returned by $\gca_{r^*}$ equals $s$ by Solo Agreement.
    For $\hat r > r^*$, by GCA Validity such a command belongs to an input trace of $\gca_{\hat r}$, which by Lines~\ref{line:uc:4}-\ref{line:uc:9} is $x$ or $x \cdot cmd''$, where $cmd''$ is the command of the proposing process and $x$ is an output trace of $\gca_{\hat r-1}$, possibly read from $S$ at Line~\ref{line:uc:4}.
    Either $cmd'=cmd''$, appended at Line~\ref{line:uc:8} when proposing to $\gca_{\hat r}$, or it belongs to $x$ and the induction hypothesis applies.

    Now let $cmd^* \in \ops{u_j}$.
    As $s$ is committed at round $r^* \leq r_0 < r$, Lemma~\hyperlink{lemmaA1}{\ref*{lemma:prefix-rounds}} gives $s \leq t$, so that $cmd^*  \notin \ops t \supseteq \ops s$.
    Applying the induction claim to $\hat r = r-1$, $cmd^*$ is appended by its own process when proposing to some $\gca_{r'}$ with $r' > r^*$, which by the choice of $r^*$ happens after $\alpha'$.
    Thus the operation instances of $\bigcup_j \ops{u_j} \cup \bigcup_j \{cmd_j\}$ are pending forever and take steps after $\alpha'$; since we are at a time $\tau' \geq \tau$, they do not conflict with each other, and the traces $s_j$ are compatible, as required.
\end{proof}

\section{Discussion and Future Work}
\label{sec:conclusion}

%
%
%

%
%

\paragraph{Boosting progress with contention managers.}
While conflict-freedom offers a flexible middle ground between obstruction-freedom and stronger liveness guarantees, some applications may require progress even under concurrent conflicting operations.
In such a setting, the universal construction (Algorithm~\ref{alg:UnivConstructionGCAV2}) can be extended with a \textit{contention manager} that resolves conflicts by enforcing an order on competing commands.
There is a rich literature on contention management and conflict resolution strategies~\cite{Guerraoui05CMTheory, Schneider09Bounds, agarwal89exponential}.
These approaches either rely on stronger synchronization primitives, such as consensus objects~\cite{Her91}, or relax liveness guarantees enabling adaptive techniques such as exponential back-off~\cite{agarwal89exponential}.

\paragraph{Bounded conflict-freedom.}
One can turn our universal constructions into \emph{bounded} (weakly) conflict-free, by introducing an upper bound $b$ (a polynomial in $n$) on the number of steps a process can take in the absence of conflicts before committing its operation. 
Bounded versions of our Algorithms~\ref{alg:UnivConstructionGCA} and~\ref{alg:UnivConstructionGCAV2} can be built by making the processes, in every attempt to commit a value, to share the outcomes of the last GCA instance they accessed and to adopt directly the outcome of the most recently accessed instance and then proceed directly to the next one (updating $r$ in lines~\ref{line:uc:6} and~\ref{line:cfuc:inc}, resp.).   %
The resulting algorithm would ensure also that an operation terminates if it takes at most $b$ steps in isolation.  
Hence, once a process has been given a chance to run solo for $b$ steps, operations invoked beforehand are included in $\asymp$ only if they continue taking steps---processes that failed cannot affect progress anymore.

%
%
%

\paragraph{Beyond interval contention, static conflicts and commutativity.}
One might be tempted to strengthen our progress condition to only account for \emph{active} conflicts, i.e., for conflicting operations that concurrently take steps. 

The resulting progress condition of \emph{conflict-obstruction-freedom (COF)}~\cite{cof}
guarantees that an operation completes if it eventually executes without step contention with conflicting operations.
This is a very attractive property, as it naturally takes care of conflicts inflicted by faulty processes. 
%
%
Unfortunately, via a reduction to the consensus impossibility~\cite{flp,LA87}, it is shown in~\cite{cof} that, in a system of three or more processes, there are no COF read-write universal constructions.
This opens an intriguing question on the synchronization power~\cite{Her91} of a COF universal construction.

Conflicts can be defined dynamically: two operations conflict only if they do not commute in any ``current'' object's state, i.e., a state derived from a possible linearization of the current execution~\cite{dynamicConf}, yielding a dynamic notion of conflict-freedom.
Also, one may account for semantic relations beyond commutativity.
Prior work shows that overwriting operations do not require synchronization~\cite{pram}; accordingly, operations could be treated as non-conflicting when they commute or when one overwrites the other. 
Generalizing conflict-free universal constructions to these relaxations of the conflict relation is left for future work.

%

\section{Acknowledgments}
This work was supported by the CHIST-ERA European initiative
(grant CHIST-ERA-22-SPiDDS-05, REDONDA project), and by
Agence Nationale de la Recherche (grant ANR-23-CHR4-0009).

\bibliographystyle{ACM-Reference-Format}

\bibliography{references}

\appendix

\section{Full Proofs} \label{sec:proofs}

\subsection{Relationships Among Liveness Conditions}

\ProgressHierarchy*
\begin{proof}
Fix an infinite execution $\alpha$ and let $\Phi$ be any operation instance invoked by a correct process in $\alpha$.

\textit{(1) wait-freedom $\implies$ conflict-freedom}.
Assume the implementation is wait-free.
By definition, every operation invoked by a correct process completes in every infinite execution.
Thus, $\Phi$ completes in $\alpha$.
In particular, $\Phi$ completes whenever $\Phi$ is eventually step-contention-free in $\alpha$, and also whenever $\alpha$ is eventually conflict free.
Therefore, the implementation is conflict-free.

\textit{(2) conflict-freedom $\implies$ weak conflict-freedom.}
Assume the implementation is conflict-free.
%
%
If $\Phi$ is step-contention-free in $\alpha$, then $\Phi$ completes.
If $\alpha$ is eventually conflict-free, then every operation instance invoked by any correct process---and in particular $\Phi$---completes.
Hence, any correct process completes all of its operation instances.

\textit{(3) weak conflict-freedom $\implies$ obstruction-freedom.}
Assume the implementation is weakly conflict-free.
By definition, $\Phi$ completes if it is eventually step-contention-free in $\alpha$, which is precisely the definition of obstruction-freedom.

Since $\alpha$ and $\Phi$ were arbitrary, the hierarchy follows.
\end{proof}

\CFDegenerate*
\begin{proof}
Fix an arbitrary infinite execution $\alpha$.
Recall that $\alpha$ is \emph{eventually conflict-free} if it has a suffix in which no two concurrent operation instances $\Phi=(o,i)$ and $\Phi'=(o',j)$ satisfy $o \asymp o'$.

\textit{Case $\asymp = O \times O$.}
If $\alpha$ is eventually conflict-free, there exists a suffix $\alpha'$ with no concurrent operation instances; otherwise any pair of concurrent operations would conflict.
Thus, every operation instance invoked by a correct process is eventually step-contention-free in $\alpha$ (it either completes or runs solo in $\alpha'$).
By definition of conflict-freedom, an operation instance must complete whenever either (I) it is eventually step-contention-free, or (II) the execution is eventually conflict-free.
Since (II) implies (I), conflict-freedom reduces to requiring operations to complete under (I) alone, which is exactly obstruction-freedom.
The same simplification applies to weak conflict-freedom.

\textit{Case $\asymp = \emptyset$.}
Since no pair of operations conflict, $\alpha$ is eventually conflict-free. 
If the implementation is conflict-free, every operation instance invoked by a correct process completes in $\alpha$, which is exactly wait-freedom.
The converse follows from Proposition~\ref{prop:progress-hierarchy}.

If the implementation is weakly conflict-free, some process completes all of its operations in $\alpha$, which is exactly lock-freedom.
Conversely, assume the implementation is lock-free. 
Let $\Phi$ be an operation instance of a correct process which is eventually step-contention-free in $\alpha$.
If $\Phi$ did not complete, there would exist a suffix in which $\Phi$ is pending and no other operation takes steps, contradicting lock-freedom. 
Hence every eventually step-contention-free operation completes.
Moreover, lock-freedom ensures that some process completes all of its operations in $\alpha$. 
Therefore, the implementation is weakly conflict-free.

\end{proof}

\subsection{Well-Defined Return Values on Traces}

\RetTraceEq*
\begin{proof}
    From the definition of a trace \cite{aalbersberg1988theory}, it suffices to consider the case where $s=x\ccdot o \ccdot o' \ccdot x'$ and $s' = x \ccdot o' \ccdot o \ccdot x'$, with $x,x' \in O^*$ and $o' \not\asymp o$, since equivalent schedules differ by a finite sequence of swaps of adjacent commuting operation occurrences.
    If $o$ and $o'$ are occurrences of the same operation, then swapping them does not change the schedule, and the claim is immediate.
    Assume henceforth $o$ and $o'$ are occurrences of different operations.
    Let $q_x=\st(\sigma^*(x,q_0)[|x|])$ be the state reached after executing $x$.
    Write $(r,q)=\sigma(o,q_x)$ and $(r',q')=\sigma(o',q)$.
    By definition of $\sigma^*$, $\sigma^*(s,q_0) = \sigma^*(x,q_0) \cdot (r,q) \cdot (r',q') \cdot \sigma^*(x',q')$.
    Since $o \not\asymp o'$, there exists a state $q''$ such that $\sigma(o',q_x)=(r',q'')$ and $\sigma(o,q'')=(r,q')$.
    Therefore 
    \begin{equation}
        \sigma^*(s',q_0) = \sigma^*(x,q_0) \cdot (r',q'') \cdot (r,q') \cdot \sigma^*(x',q')
    \end{equation}
    Recall that the occurrence $o$ in appears at position $|x|+1$ in $s$ and at position $|x|+2$ in $s'$. 
    Hence, by Equation (1) and the definition of $\ret^*$, $\ret^*(o,s)=\ret(\sigma^*(s,q_0)[|x|+1])=\ret(\sigma^*(s',q_0)[|x|+2])=\ret(o,s')$.
    The same argument applies to $o'$.
\end{proof}

\RetTracePrefix*
\begin{proof}
Since $s \leq t$, there exists $u\in O^*/\!\sim$ such that $t=s \cdot u$.
Fix an occurrence $o^{(i)}$ in $s$.
Choose representative schedules $\hat s \in s$ and $\hat u \in u$.
Then, $\hat s \cdot \hat u \in t$. 
By definition of $\sigma^*$, $\sigma^*(\hat s,q_0)$ is a prefix of $\sigma^*(\hat s \ccdot \hat t,q_0)$.
Thus, the occurrence $o^{(i)}$ has the same return value in
$\hat s$ and $\hat s \cdot \hat u$.
By the definition of return values on traces,
$\ret^*(o^{(i)},s)=\ret^*(o^{(i)},\hat s)=\ret^*(o^{(i)},\hat s \cdot \hat u)=\ret^*(o^{(i)},t)$.
%
%
\end{proof}

\subsection{Correctness of the Weakly Conflict-free Universal Construction}
\label{sec:appendix:weak_uc}

In what follows, we first state a technical lemma that applies to both Algorithm~\ref{alg:UnivConstructionGCA} and Algorithm~\ref{alg:UnivConstructionGCAV2}, then we prove that the universal construction in Section~\ref{sec:weak_uc} is both linearizable and weakly conflict-free.

\hypertarget{lemmaA1}{}
\begin{lemma}
\label{lemma:prefix-rounds}
    Fix an execution of Algorithm~\ref{alg:UnivConstructionGCA} or Algorithm~\ref{alg:UnivConstructionGCAV2}.
    If $(s,\True)$ is returned by $\gcapropose{r}{\cdot}$ for some round $r\geq 1$, then for every round $r'\geq r$, every trace $t$ returned by $\gcapropose{r'}{\cdot}$ satisfies $s \leq t$.
\end{lemma}
\begin{proof}
    We prove the claim by induction on $r' \geq r$.
    \textit{Base case $(r'=r)$.}
    By the Adoption property of GCA, if $(s,\True)$ is returned in round $r$, then every trace $t$ returned by $\gcapropose{r}{\cdot}$ satisfies $s \leq t$.
    \textit{Inductive step.}
    Assume the claim holds for some round $r'\geq r$.
    We prove it for round $r'+1$.
    Consider any trace $x$ proposed to $\gca_{r'+1}$.
    In Algorithm~\ref{alg:UnivConstructionGCA}, $x$ extends a trace that its proposer either returned from $\gcapropose{r'}{\cdot}$ or read, with round $r'$, from the array $S$ at Line~\ref{line:uc:4}, possibly followed by commands appended at Line~\ref{line:uc:8}.
    The same holds in Algorithm~\ref{alg:UnivConstructionGCAV2}, with Lines~\ref{line:cfuc:6} and~\ref{line:cfuc:11} in place of Lines~\ref{line:uc:4} and~\ref{line:uc:8}.
    In the latter case, some process wrote that trace to $S$ after committing it at round $r'$.
    In both cases, $x$ extends a trace returned by $\gcapropose{r'}{\cdot}$, and therefore, by the induction hypothesis, it extends $s$. Hence $s \leq x$.
    Thus, every input trace proposed to $\gcapropose{r'+1}{\cdot}$ extends $s$.
    By GCA Common Prefix, every trace returned by $\gca_{r'+1}$ extends the GLB of the inputs, which itself extends $s$.
    Therefore $s \leq t$ for every trace $t$ returned in round $r'+1$.
\end{proof}

Recall that a history $H$ is linearizable if
it has  a completion $\bar H$ and a legal sequential history $S$ such that \emph{(i)} $\bar H |_i = S|_i$ for every process $i$, and \emph{(ii)}~the real-time order of $H$ is preserved in $S$. 
An implementation is linearizable if every execution produces a linearizable history.
As linearizability is a safety property for deterministic objects~\cite{Lyn96}, to show that an implementation is linearizable, it is sufficient to show that every finite history it produces is linearizable.

\WeakUCLin*
\begin{proof}
Let $\alpha$ be a finite execution of Algorithm~\ref{alg:UnivConstructionGCA} and $H(\alpha)$ be its associated history. 
%
%
We show that $H$ is linearizable. 
%

According to Line~\ref{line:uc:3} in Algorithm~\ref{alg:UnivConstructionGCA}, every invocation of some operation $\Phi$ creates a unique corresponding command, written hereafter $cmd(\Phi)$.
Operation $\Phi$ completes only after committing some trace $s$ from a call to $\gcapropose{r}{\cdot}$, with $cmd(\Phi) \in \ops s$ and $r\geq 1$.
The value returned from invoking $\Phi$ is computed in Line~\ref{line:uc:11} and equals $\ret^*(cmd(\Phi),s)$.
Let $r^*$ be the largest round in which some trace is committed in $\alpha$.
By Adoption, there is a single such trace $t$ at $r^*$.

We show the existence of a representative schedule $\hat t \in t$ such that for any two distinct operations $\Phi_1 \preceq_{H} \Phi_2$, $cmd(\Phi_1) \preceq_{\hat t} cmd(\Phi_1)$.
First of all, by Lemma~\hyperlink{lemmaA1}{\ref*{lemma:prefix-rounds}}, $cmd(\Phi_1)$ and $cmd(\Phi_2)$ are always in a representative of $t$.
The representative $\hat{t}$ is built as follows:
For any round $k$ at which some trace is committed, let $t_k$ be this unique trace.
By convention, fix $t_0=\epsilon$ and $t_{k+1}=t_k$ if no trace is committed at round $k+1$.
By Lemma~\hyperlink{lemmaA1}{\ref*{lemma:prefix-rounds}}, there exists a family $(s_k)_{k \in [0,r^*]}$ such that $t_{k+1}=t_k \cdot s_k$.
It follows that $t=t_{r^*}$ and $t=s_1 \cdot \ldots \cdot s_{r^*}$.
From what precedes, for every process $i$ and operation $\Phi \in H$ invoked by $i$, there exists a unique $t_k$ such that $i$ returns from its invocation of $\Phi$ after committing $t_k$.
If the call does not return, by convention we associate $t_{r^*}$ to $cmd(\Phi)$.
Our representative $t$ is the concatenation (in order) of the representatives of each $s_k$.
That is, denoting $\hat{s}_k$ some representative of $s_k$, $\hat{t}=\hat{s}_1 \ldots \hat{s}_{r^*}$.
We observe that if $\Phi_1 \preceq_{H} \Phi_2$, denoting $t_k$ and $t_{k'}$ traces associated to commands $cmd(\Phi_1)$ and $cmd(\Phi_2)$, then $k < k'$, as required.

Define the sequential history $S_{\hat t}$ as the alternating invocation/response sequence obtained by executing $\hat t= cmd(\Phi_1)\cdots cmd(\Phi_{|\hat t|})$ from $q_0$ using $\sigma^*$.
%
By definition of $\sigma^*$, $S_{\hat t}$ is legal.
By construction, $\preceq_{H} \subseteq \preceq_{S_{\hat t}}$.

Now fix any completed operation instance $\Phi$.
Let $s$ be the trace obtained in Line~\ref{line:uc:10}. 
By Lemma~\hyperlink{lemmaA1}{\ref*{lemma:prefix-rounds}}, $s \leq t$.
Thus by Lemmas~\ref{lemma:ret-trace-prefix} and~\ref{lemma:ret-trace-eq}, $\ret^*(cmd(\Phi),s)=\ret^*(cmd(\Phi),t)=\ret^*(cmd(\Phi),\hat t)$.
By definition of $S_{\hat t}$, $\ret^*(cmd(\Phi),\hat t)$ is exactly the response of $\Phi$ in $S_{\hat t}$.
%

It remains to define a completion $\bar{H}$ of $H$.
To this end, consider the pending operations in $H$.
For each such operation $\Phi$, if $cmd(\Phi) \in \ops{t}$, we append to $H$ a matching response event returning $\ret^*(cmd(\Phi),t)$.

By construction, for every process $i$, $\bar{H}|i = S_{\hat{t}}|i$.
Also, $S_{\hat{t}}$ is legal and $\preceq_{H} \subseteq \preceq_{S_{\hat t}}$.
%
Thus, $S_{\hat{t}}$ is a linearization of $H$.
\end{proof}

\WeakUCWCF*
\begin{proof}
Let $\alpha$ be some infinite execution that is eventually conflict-free.
By definition, there exists a point in time $\tau$ after which no two pending operations are conflicting.
%
Then, for the sake of contradiction, suppose no correct process completes all its operations in $\alpha$.
It follows that there exists a time $\tau' \geq \tau$ after which only correct processes take steps and all their operations are pending forever.
Let $r^*$ be the largest round reached by any process by time $\tau'$.

Assume that the following invariant holds:
(I) for every $r>r^*$, only compatible traces are proposed to $\text{GCA}_{r}$.
%
%
By invariant (I), all proposals to $\gca_{r^*+1}$ are compatible.
By Commitment, there exists a correct process and an operation instance $\Phi$ such that this process invokes $\gcapropose{r^*+1}{s}$, with $cmd(\Phi) \in s$, and returns $(x, \True)$ such that $s \leq x$.
Thus, the process invoking $\Phi$ exits the while loop at Line~\ref{line:uc:5}, completing its operation---a contradiction.

In what follows, we prove that invariant (I) holds.
Given some correct process $i$, let $cmd_i$ be the command containing the forever-pending operation of $i$. 
Consider some round $r > r^*$.
Let $t_i$ and $s_i$ be the traces $i$ retrieves from $\gca_{r-1}$ and proposes to $\gca_{r}$, respectively.
According to Lines~\ref{line:uc:7} and \ref{line:uc:8}, $s_i = t_i \cdot cmd_i$ if $cmd_i \notin \ops {t_i}$, and $s_i=t_i$ otherwise.
We establish that $\mathcal{S} = \bigcup_i \{s_i\}$ is compatible. 
By the Convergence property, $\mathcal{T}= \bigcup_i \{t_i\}$ is compatible.
Let $t$ be the common prefix of these traces, that is $t = \bigsqcap \mathcal{T}$.
Given a correct process $i$, $u_i$ is the (unique) trace such that $t_i = t \cdot u_i$.
By GCA Commitment and Common Prefix, $u_i$ contains only commands from pending operations.
Indeed, if some operation $\Phi$ has already returned, then $cmd(\Phi)$ must be in some committed trace.
By Lemma~\hyperlink{lemmaA1}{\ref*{lemma:prefix-rounds}}, every such trace prefixes $t_i$.
Thus, $cmd(\Phi)$ is included in the prefix $t$.
By definition, every command $cmd_i$ must also be pending.
Thus, the set $\bigcup_i \ops{u_i} \cup \bigcup_i \{cmd_i\}$ only contains pending operations.
Since we are at a time $\tau' > \tau$, these operations do not conflict between each other.
%
It follows that $\mathcal{S} = \bigcup_i \{t_i \cdot cmd_i : cmd_i \notin \ops{t_i}\} \cup \bigcup_i\{t_i : cmd_i \in \ops {t_i}\}$ only contains compatible traces, as required.

It remains to show that for any infinite execution $\beta$ and any operation instance $\Phi$ of a correct process $i$ that is eventually step-contention-free in $\beta$, $\Phi$ completes.
By definition, there exists a time $\tau$ after which only $\Phi$ takes steps in $\beta$.
Let $r^*$ be the largest round reached by any process by time $\tau$.
Since $i$ continues executing the loop in Lines~\ref{line:uc:5}-\ref{line:uc:10}, it eventually invokes $\gcapropose{r^*+1}{s}$ with $cmd(\Phi)\in \ops{s}$.
As $i$ is the only participant in $\gca_{r^*+1}$, by the Commitment property, the call returns $(s,\True)$.
Thus, the process exits the while loop and $\Phi$ completes.
\end{proof}

\subsection{Correctness of the GCA implementation}
\label{sec:appendix:gca}

We now prove the correctness of Algorithm~\ref{alg:GCAFV2}. 
Specifically, we aim to show that it satisfies all the six properties that define GCA in Section~\ref{sec:gca} and the additional property of Solo Agreement defined in Section~\ref{sec:resolution}.
Each property is captured in a separate lemma, and collectively they establish the main correctness theorem.
As introduced in Section~\ref{sec:gca}, $P \subseteq \Pi$ denotes the participating processes.

We first extend the definitions of $\comp{}$, $\bigsqcap$, and $\bigsqcup$ to admit the distinguished value $\bot$.
For every finite set $S$ of traces and $\bot$ values, we extend the operators as follows:
$\comp{S}$ and $\bigsqcup S$ are computed on $S\!\setminus\!\{\bot\}$, and $\bigsqcap S = \bot$ whenever $\bot \in S$.
By convention, if $S \setminus \{\bot\} = \emptyset$, then $\comp{S}$ holds and $\bigsqcup S = \epsilon$.

\hypertarget{lemmaA2}{}
\begin{lemma}[Adoption] \label{lemma:GCA_adoption}
If a process $i$ commits a trace $t_i$, then no process commits or adopts a trace $t_j$ such that $t_i \not\leq t_j$. 
\end{lemma}
\begin{proof}
    By contradiction, suppose that a process $i$ commits a trace $\beta_i$, while a process $j$ adopts or commits a trace $\beta_j$ with $\beta_i \not\leq \beta_j$. 

    
    Given a process $l$, let us define $B_l^c = \{b_k : (b_k,c_k) \in B_l \land b_k \neq \bot \land c_k = \True \}$.

    First of all, we observe that $\beta_i = \bigsqcap B_i^c = \bigsqcap B_i$.
    To prove this, there are two cases to distinguish:
    (Case $w_i=\False$)
    According to Line~\ref{line:gca:8}, every compatibility flag $c_k$ in $B_i$ must be raised, that is, set to $\True$.
    Consequently, process $i$ uses the first branch of the ``if'' in Line~\ref{line:gca:6} to compute the output trace $\beta_i$.
    According to the pseudo-code in this line, $\beta_i = \bigsqcap B_i^c$.
    (Case $w_i=\True$)
    In this situation, for every process $k$, every non-$\bot$ entry $a_k \in A_i$ or $b_k \in B_i$ must equal $s_i$, the input trace of $i$.
    In particular, $A_i = A_i^\mathsf{co}$, so that $i$ writes $(s_i,\True)$ into $B[i]$ in Line~\ref{line:gca:4}, and by self inclusion $(s_i,\True) \in B_i$.
    Thus, process $i$ also uses the first branch of the ``if'' in Line~\ref{line:gca:6}, and $\beta_i = \bigsqcap B_i^c = s_i = \bigsqcap B_i$.
    
    Next, we show that $B_i \subset B_j$. 
    First of all, since $B_i$ and $B_j$ are related by (snapshot) containment, it must be the case that either $B_i \subset B_j$, $B_j \subset B_i$, or $B_i = B_j$.
    (Case $B_i = B_j$)
    Then $\beta_i = \beta_j$, but $\beta_i \not\leq \beta_j$;
    a contradiction.
    (Case $B_j \subset B_i$) 
    Since $B_j \subset B_i$, $\beta_j \geq \bigsqcap B_j \geq \bigsqcap B_i = \beta_i$ (by Line~\ref{line:gca:6}, $\beta_j$ is either $\bigsqcap B_j^c$ or $b_j \in B_j$);
    a contradiction.
       
    From what precedes, $B_i \subset B_j$.
    Thus $j$ takes the first branch of the ``if'' in Line~\ref{line:gca:6} and $\beta_j = \bigsqcap B_j^c$.
    This implies that $\beta_i > \beta_j$.

    If $\beta_i > \beta_j$, then there exists $(b_z,\True) \in B_j \setminus B_i$, with $\beta_i > b_z$.
    Indeed, since $\beta_i=\bigsqcap B_i^c$ and $\beta_j=\bigsqcap B_j^c$, and all entries in $B_i^c \cup B_j^c$ are pairwise comparable (being LUBs of containment-ordered snapshot views of $A$), having $b_k\geq \beta_i$ for every $(b_k,\True)\in B_j\setminus B_i$ would imply $\beta_j\ge\beta_i$, a contradiction. 

    Now we separate in two cases, whether $w_i = \True$ or $w_i = \False$,  and show that in both cases $A_z \not\subset A_i$.
    If $w_i = \False$, by the condition on Line~\ref{line:gca:8}, since $b_z = \bot$ in $B_i$, either $a_z = \bot$ in $A_i$ or $a_z \not\leq \beta_i$; the latter is impossible, as $a_z \leq b_z < \beta_i$.
    Thus, $a_z=\bot$ in $A_i$. 
    By self inclusion on $A$, we have that $a_z=s_z$ in $A_z$, thus $A_z \not\subset A_i$.
    If $w_i=\True$, all non-$\bot$ entries in $A_i$ and $B_i$ are equal to $s_i$.
    As $a_z \leq b_z < \beta_i = s_i$, necessarily $a_z = \bot$ in $A_i$. Thus $A_z \not\subset A_i$.

    Now suppose $A_i \subset A_z$. 
    Since the flag $c_z$ is raised, $b_z = \bigsqcup A_z^\mathsf{co} = \bigsqcup A_z$.
    Moreover, all the traces in $A_z$, and thus also in $A_i$ are compatible.
    In such a case, $\beta_i \leq b_i = \bigsqcup A_i^{\mathsf{co}} = \bigsqcup A_i \leq \bigsqcup A_z$.
    However, by definition of trace $b_z$, $\beta_i > b_z$.
    A contradiction.

    As $A_i \not\subset A_z$ and $A_z \not\subset A_i$, necessarily $A_i = A_z$.
    Because the flag $c_z$ is raised, traces in $A_i$ are all compatible.
    Thus, $\beta_i \leq b_i = \bigsqcup A_i$.
    On the other hand, by definition of $b_z$, $\beta_i > b_z = \bigsqcup A_z =  \bigsqcup A_i$.
    We reach a contradiction.
\end{proof}


\begin{lemma}[Commitment]
\label{lemma:GCA_commit_input}
    In any execution where all the input traces are compatible and every participating process return, there is a process $i$ with input trace $s_i$ that commits a trace $t_i$ with $s_i \leq t_i$. 
\end{lemma}
\begin{proof}
    Since all the inputs trace are compatible, for every process $k$, $A_k = A_k^\mathsf{co}$.
    According to Line~\ref{line:gca:4}, when $k$ writes to entry $B[k]$, it writes $(\bigsqcup A_k,\True)$.
    By containment on $A$, for any two processes $j$ and $k$, either $A_k \subseteq A_j$ or the converse holds.
    It follows that $\bigsqcup A_k \leq \bigsqcup A_j$ or $\bigsqcup A_k \geq \bigsqcup A_j$.
    Thus, all the traces written in $B$ are \emph{comparable} with respect to $\leq$.

    By containment on $B$, all the snapshots $(B_l)_{l \in P}$ taken in Line~\ref{line:gca:5} are comparable w.r.t. $\subseteq$.
    Since every entry in $B$ has its flag set to $\True$, when executing Line~\ref{line:gca:6} processes all take the first branch of the ``if''.
    It follows that all the traces $\{\beta_l  : l \in P\}$ computed in Line~\ref{line:gca:6} are comparable w.r.t. $\leq$.
    Let $\beta$ be the smallest such trace.  

    Let $P_\beta \subseteq P$ be the subset of processes that output $\beta$ and whose input trace is a prefix of $\beta$.
    By definition of $\beta$, and since the traces written in $B$ are comparable, $\beta$ is the smallest trace in some $B_l$, so at least one process $q$ writes $\beta$ into $B$.
    For such a process $q$, it is true that $\beta = b_q = \bigsqcup A_q \geq s_q$, and $q$ outputs $\beta$ (because $B_q$ contains $b_q$ and $\beta$ is the smallest output trace).
    Hence, $q \in P_\beta$, i.e., $P_\beta \neq \emptyset$.

    It remains to prove that at least one process in $P_\beta$ commits.
    In what follows, we reason by contradiction, supposing instead that all of them adopt $\beta$.

    We define a function $f : P_\beta \rightarrow P$ such that for every $i \in P_\beta$, $f(i)=k$ if and only if $a_k \leq \beta_i$ in $A_i$ and $b_k = \bot$ in $B_i$. 
    If multiple such $k$ exist for a given $i$, we select one of them arbitrarily.
    Notice that, due to Line~\ref{line:gca:8}, since all the processes in $P_\beta$ adopt the trace $\beta$, such a $k$ must always exist.
    Thus function $f$ is well-defined.

    Fix some $i \in P_\beta$. 
    There exists $f(i) \in P$ such that $a_{f(i)} \leq \beta$ and $b_{f(i)} = \bot$ in $B_i$.
    By containment on $B$, we have $B_i \subset B_{f(i)}$.
    Then, $\beta = \bigsqcap B_i \geq \bigsqcap B_{f(i)} \geq \beta$. 
    Therefore $\bigsqcap B_{f(i)} = \beta$ and, since the input trace of $f(i)$ is $a_{f(i)} \leq \beta$, ${f(i)} \in P_\beta$.

    Repeating this argument $n+1$ times, results in a chain $B_i \subset B_{f(i)} \subset \dots \subset B_{f^{n+1}(i)}$. 
    Each $f^k(i)$ in this chain must be distinct and must belong to $P_\beta$ for  $1\leq k\leq n+1$.
    However, this yields a contradiction, since $|P| \leq n+1$.
    Thus, function $f$ cannot generate a strictly increasing chain of length $n+1$ within $P_\beta$.

    Therefore, at least one process in $P_\beta$ commits. 
    By definition of $P_\beta$, such a process commits the trace $\beta$ which extends its input.

\end{proof}

\begin{lemma}[Convergence]
Output traces are mutually compatible.
\end{lemma}
\begin{proof}
    From Line~\ref{line:gca:6}, each output trace is the GLB of a subset of (non-$\bot$) entries in $B$. 
    Thus, to prove that the output traces are compatible, it suffices to show that all the traces written into $B$ are compatible.

    For the sake of contradiction, suppose two processes $i$ and $j$ write non-compatible traces $b_i$ and $b_j$ in $B$.
    According to Line~\ref{line:gca:4}, $b_i= \bigsqcup A_i^\mathsf{co}$ and  $b_j= \bigsqcup A_j^\mathsf{co}$, where $A_i^\mathsf{co}$ and $A_j^\mathsf{co}$ are compatible sets by construction. 
    Thus, for $b_i$ and $b_j$ to be non-compatible there exist $a_k^\mathsf{co} \in A_i^\mathsf{co}$ and $a_l^\mathsf{co} \in A_j^\mathsf{co}$ which are not compatible.
    From Line~\ref{line:gca:3}, 
    \begin{align*}
    a_k^{\mathsf{co}} & = \bigsqcap(\{a_k\}\cup\{a_p \in A_i: \neg\comp{a_p, a_k} \}) \\
    a_l^{\mathsf{co}} & = \bigsqcap(\{a_l\}\cup\{a_p \in A_j: \neg\comp{a_p, a_l} \})
    \end{align*}
    Thus, $a_k \in A_i$ and $a_l \in A_j$ must themselves be non-compatible.
    Moreover, for $a_k^{\mathsf{co}}$ and $a_l^{\mathsf{co}}$ to not be compatible, $a_l \notin  \{a_p \in A_i : \neg\comp{a_p, a_k} \}$.
    This implies $a_l \notin A_i$.
    Similarly, $a_k \notin \{a_p \in A_j : \neg\comp{a_p, a_l} \}$, leading to  $a_k \notin A_j$.

    Hence, $a_k \in A_i \setminus A_j$ and $a_l \in A_j \setminus A_i$, implying that $A_i \nsubseteq A_j$ and $A_j \nsubseteq A_i$. 
    A contradiction, as $A$ is an atomic snapshot object and its views must be related by containment. 
    Therefore, all the traces in $B$ are compatible.
    That is, there exists a trace $z$ such that for any trace $b_l$, $b_l \leq z$.
    As for any process $i$, $b_i \leq z$ and $\bigsqcap\{b_k : (b_k,c_k)  \in B_i \land b_k \neq \bot \land c_k=\True\} \leq z$,
    the set of output traces, as computed in Line~\ref{line:gca:6}, is also compatible.
\end{proof}

\begin{lemma}[Common Prefix]
\label{lemma:GCA_common_prefix}
Let $\alpha = \bigsqcap_{i\in P} s_i$ be the common prefix of all the input traces.
For any process $j \in P_r$, if $j$ adopts or commits trace $\beta_j$, then $\alpha \leq \beta_j$.
\end{lemma}
\begin{proof}
    A process $i \in P$ writes its input trace $s_i$ in entry $A[i]$.
    Then it takes a snapshot $A_i=(a_1, \ldots, a_n)$ of $A$ in Line~\ref{line:gca:2}.
    At Line~\ref{line:gca:3}, $a_k^{\mathsf{co}}$ is set to the value $\bigsqcap(\{a_k\} \cup \{a_j : \neg\comp{a_j, a_k} \})$.
    By assumption, $\alpha \leq a_k$, for any participating process $k \in P$.
    Consequently, $\alpha \leq a_k^{\mathsf{co}}$ and $\alpha \leq \bigsqcup A_i^\mathsf{co}$.
    Then, in Line~\ref{line:gca:5}, process $i$ writes $A_i^\mathsf{co}$ in entry $B[i]$.
    Further, it takes a snapshot $B_i=((b_1,c_1), \ldots, (b_n,c_n)$ of $B$.
    From what precedes, for every entry $b_k$ in $B_i$, $\alpha \leq b_k$.
    It follows that the output trace $\beta_i$ computed in Line~\ref{line:gca:6} also satisfies $\alpha \leq \beta_i$.
\end{proof}

\begin{lemma}[Weak Agreement]
\label{lemma:GCA_weak_agreement}
    If all input traces are equal, then no process adopts.
\end{lemma}
\begin{proof}
    Pick some participating process $j \in P$.
    We have that for any entry $a_k \in A_j$, $a_k =s$. 
    Then, $A_j = A_j^\mathsf{co}$ and $\bigsqcup A_i^\mathsf{co} = s$. 
    Thus, process $j$ writes $(s,\True)$ in $B[j]$.
    
    Now, let $i \in P_r$ be a process that returns from its call.
    From what precedes, for every non-$\bot$ entry in $B_i$, this entry equals $(s,\True)$. 
    Thus, $\beta_i = s$ and $w_i = \True$. 
    According to Line~\ref{line:gca:8}, process $i$ commits trace $s$.
\end{proof}

\hypertarget{lemmaA7}{}
\begin{lemma}[Validity]
\label{lemma:GCA_validity}
    Every output trace contains only operations from the input traces.
\end{lemma}
\begin{proof}
    According to Section~\ref{sec:traces}, given some set of traces $S$, $\ops{\bigsqcap S} \subseteq \bigcap_{s \in S} \ops{s}$.
    Similarly, $\ops{\bigsqcup S} \subseteq \bigcup_{s \in S} \ops{s}$.
    As a consequence, in light of the pseudo-code of Algorithm~\ref{alg:GCAFV2}, every output trace is built with commands from the input traces.
\end{proof}

\hypertarget{lemmaA8}{}
\begin{lemma}[Solo Agreement]
\label{lemma:GCA_soloagg}
    If a process returns before any other process participates, every output trace equals its own.
\end{lemma}
\begin{proof}
    Let $\alpha'$ be some prefix of the execution in which some process $j$ returns from its call before any other process participates.
    Every entry of $A$ and $B$ other than $A[j]$ and $B[j]$ thus equals $\bot$ at the end of $\alpha'$, and $A[j]=s_j$ by Line~\ref{line:gca:1}.
    We have that all entries in $A_j$, except $A_j[j]=s_j$, equal $\bot$.
    Thus, $A_j^\mathsf{co}=A_j$ and process $j$ writes $(s_j,\True)$ in $B[j]$ at Line~\ref{line:gca:4}.

    Next, we show that every pair $(b_k,\True)$ written in $B$ satisfies $s_j \leq b_k$.
    Consider a process $k \neq j$ that writes $(b_k,\True)$ in $B[k]$.
    Process $k$ participates after $j$ returns, so that its snapshot $A_k$ contains $a_j = s_j$.
    Since its flag is $\True$, $A_k = A_k^\mathsf{co}$ by Line~\ref{line:gca:4}, and in particular $a_j^\mathsf{co} = s_j$.
    Hence $s_j \leq \bigsqcup A_k^\mathsf{co} = b_k$.

    Now, let $i \in P_r$ be a process that returns from its call.
    Process $i$ takes its snapshot $B_i$ after $j$ writes $(s_j,\True)$ in $B[j]$, so that $(s_j,\True) \in B_i$.
    Consequently, process $i$ uses the first branch of the ``if'' in Line~\ref{line:gca:6}, and $\beta_i = \bigsqcap\{b_k : (b_k,c_k) \in B_i \land c_k = \True\}$.
    From what precedes, $s_j$ belongs to this set and is a prefix of each of its elements.
    Thus, $\beta_i = s_j$.
\end{proof}

From what precedes, we may now establish the main result.

\GCAWorks*

\begin{proof}
    In light of Lemmas~\hyperlink{lemmaA2}{\ref*{lemma:GCA_adoption}}-\hyperlink{lemmaA7}{\ref*{lemma:GCA_validity}}, all the properties of GCA hold.
    It is also clear from its pseudo-code that Algorithm~\ref{alg:GCAFV2} is wait-free.
\end{proof}

\end{document}